\documentclass{article}
\usepackage[preprint,nonatbib]{nips_2018}
\usepackage[utf8]{inputenc} 
\usepackage[T1]{fontenc}    
\usepackage[hidelinks]{hyperref}       
\usepackage{url}            
\usepackage{booktabs}       
\usepackage{amsfonts}       
\usepackage{nicefrac}       
\usepackage{microtype}      

\usepackage{amsmath}
\usepackage{amssymb}
\usepackage{amsthm}
\usepackage{mathtools}
\usepackage{thmtools}
\usepackage{bbm}
\usepackage[dvipsnames]{xcolor}
\usepackage{verbatim}

\usepackage{tikz}
\usepackage{pgfplots}

\usepackage{subcaption}

\pgfplotsset{compat=newest}
\pgfplotsset{plot coordinates/math parser=false}

\usetikzlibrary{decorations.pathreplacing}
\usetikzlibrary{patterns}
\usepgfplotslibrary{groupplots}

\newlength{\figwidth}
\newlength{\figheight}
\definecolor{gray1}{gray}{0.0}
\definecolor{gray2}{gray}{0.25}
\definecolor{gray3}{gray}{0.5}
\definecolor{gray4}{gray}{0.7}
\definecolor{gray5}{gray}{0.9}

\newlength{\prel}

\pgfplotsset{
  title style = {font=\small}
}

\DeclareUnicodeCharacter{00A0}{ } 

\declaretheorem[Refname={Theorem,Theorems}]{theorem}
\numberwithin{theorem}{section}
\declaretheorem[numberlike=theorem,Refname={Corollary,Corollaries}]{corollary}
\declaretheorem[numberlike=theorem,Refname={Proposition,Propositions}]{proposition}
\declaretheorem[style=definition,numberlike=theorem,Refname={Definition,Definitions}]{definition}
\declaretheorem[style=definition,numberlike=theorem,Refname={Example,Examples}]{example}
\declaretheorem[style=definition,numberlike=theorem,Refname={Remark,Remarks}]{remark}

\DeclarePairedDelimiter\abs{\lvert}{\rvert} 
\DeclarePairedDelimiter\norm{\lVert}{\rVert} 

\newcommand{\R}{\mathbb{R}} 
\newcommand{\N}{\mathbb{N}} 

\DeclareMathOperator*{\argmin}{arg\,min}    

\newcommand{\polspace}{\pi} 

\title{A Bayes--Sard Cubature Method}

\author{
  Toni Karvonen \\
  Aalto University, Finland \\
  \And
  Chris. J. Oates \\
  Newcastle University, UK \\
  \And
  Simo S\"{a}rkk\"{a} \\
  Aalto University, Finland \\
}

\begin{document}

\maketitle

\begin{abstract}
This paper focusses on the formulation of numerical integration as an inferential task.
To date, research effort has largely focussed on the development of \emph{Bayesian cubature}, whose distributional output provides uncertainty quantification for the integral.
However, the point estimators associated to Bayesian cubature can be inaccurate and acutely sensitive to the prior when the domain is high-dimensional.
To address these drawbacks we introduce \emph{Bayes--Sard cubature}, a probabilistic framework that combines the flexibility of Bayesian cubature with the robustness of classical cubatures which are well-established.
This is achieved by considering a Gaussian process model for the integrand whose mean is a parametric regression model, with an improper flat prior on each regression coefficient.
The features in the regression model consist of test functions which are guaranteed to be exactly integrated, with remaining degrees of freedom afforded to the non-parametric part.
The asymptotic convergence of the Bayes--Sard cubature method is established and the theoretical results are numerically verified.
In particular, we report two orders of magnitude reduction in error compared to Bayesian cubature in the context of a high-dimensional financial integral.
\end{abstract}

\section{Introduction}

This paper considers the numerical approximation of an integral $I(f^\dagger) \coloneqq \int_D f^\dagger \mathrm{d}\nu$ of a continuous integrand $f^\dagger : D \rightarrow \mathbb{R}$ against a Borel distribution $\nu$ defined on a domain $D \subseteq \mathbb{R}^d$.
The approximation of such integrals is a fundamental task in applied mathematics, statistics and machine learning.
Indeed, the scope and ambition of modern scientific and industrial computer codes is such that the integrand $f^\dagger$ can often represent the output of a complex computational model.
In such cases the evaluation of the integrand is associated with a substantial resource cost and, as a consequence, the total number of evaluations will be limited.
The research challenge, in these circumstances, manifests not merely in the design of a cubature method but also in the assessment of the associated error.

The (generalised) \emph{Bayesian cubature} (BC) method \cite{Larkin1972,OHagan1991,Minka2000} provides a statistical approach to error assessment.
In brief, let $\Omega$ be a probability space and consider a hypothetical Bayesian agent who represents their epistemic uncertainties in the form of a stochastic process $f \colon D \times \Omega \rightarrow \mathbb{R}$.
This stochastic process must arise from a Bayesian regression model and be consistent with obtained evaluations of the true integrand, typically provided on a discrete point set $\{x_i\}_{i=1}^n \subset D$; that is $f(x_i , \omega) = f^\dagger(x_i)$ for almost all $\omega \in \Omega$.
The stochastic process acts as a stochastic model for the integrand $f^\dagger$, implying a random variable $\omega \mapsto \int_D f(\cdot,\omega) \mathrm{d}\nu$ that represents the agent's epistemic uncertainty for the value of the integral $I(f^\dagger)$ of interest.

The output of a (generalised) Bayesian cubature method is the law of the random variable \sloppy{${\omega \mapsto \int_D f(\cdot,\omega) \mathrm{d}\nu}$}.
The mean of this output provides a point estimate for the integral, whilst the standard deviation indicates the extent of the agent's uncertainty regarding the integral.
The properties of this probabilistic output have been explored in detail for the case of a centred Gaussian stochastic process (the \emph{standard} Bayesian cubature method):
In certain situations the mean has been shown to coincide with a kernel-based integration method~\cite{Oettershagen2017} that is rate-optimal \cite{Bach2017,Briol2017}, robust to misspecification of the agent's belief \cite{Kanagawa2016,Kanagawa2017} and efficiently computable \cite{Oettershagen2017,Karvonen2018}.
The non-Gaussian case and related extensions have been explored empirically in \cite{Osborne2012,Gunter2014,Oates2017,Chai2018}.
The method has also been discussed in connection with \emph{probabilistic numerics}; see \cite{Diaconis1988,Hennig2015,Cockayne2017} for general background.

However, it remains the case that non-probabilistic numerical integration methods, such as Gaussian cubatures~\cite{Gautschi2004} and quasi-Monte Carlo methods~\cite{Hickernell1998}, are more widely used, due in part to how their ease-of-use or reliability are perceived.
This is despite the well-known fact that the trapezoidal rule and other higher-order spline methods~\cite{Davis2007} can be naturally cast as Bayesian cubatures if the stochastic process $f$ is selected suitably~\cite{Diaconis1988}.
It is also known that Gaussian cubature can be viewed as a special (in fact, degenerate) case of a kernel method~\cite{Sarkka2016,Karvonen2017a}.
However, no overall framework to derive probabilistic analogies of popular cubatures, with corresponding ease-of-use and reliability, has yet been developed.

This paper argues that the perceived performance gap between probabilistic and non-probabilistic methods should be reconsidered.
To this end, we consider a non-parametric Bayesian regression model augmented with a parametric component. 
The features in the parametric component, that is the pre-specified finite set of basis functions, will be denoted $\pi$. 
Then, an improper uniform prior limit on the regression coefficients (see~\cite{OHagan1978} and \cite[Sec.\ 2.7]{Rasmussen2006}) is studied. 
This gives rise to \emph{Bayes--Sard cubature}\footnote{Our terminology is motivated by resemblance to the (non-probabilistic) method of Sard~\cite{Sard1949} for selecting weights for given $n$ nodes by fixing a polynomial space of degree $m < n$ on which the integration rule must be exact and disposing of the remaining $n-1-m$ degrees of freedom by minimising an appropriate error functional. See also \cite{Schoenberg1964} and \cite{Larkin1970}.} (BSC), which differs at a fundamental level to standard Bayesian cubature, in that the functions in $\pi$ are now exactly integrated.
The extension is similar to that proposed in 1974 by Larkin~\cite{Larkin1974}, and non-probabilistic versions have appeared independently in~\cite{Bezhaev1991,DeVore2017} in the context of interpolation with conditionally positive definite kernels and optimal approximation in reproducing kernel Hilbert spaces.
Our contributions therefore include (i) establishing a coherent and comprehensive Gaussian process framework for Bayes--Sard cubature; (ii) rigorous study of convergence and conditions that need to be established on $\pi$; (iii) empirical experiments that demonstrate improved accuracy in high dimensions and robustness to misspecified kernel parameters compared to Bayesian cubature; and (iv) the important observation that, when the dimension of the function space $\pi$ matches the number of cubature nodes, the Bayes--Sard cubature method can be used to endow \emph{any} cubature rule with a meaningful probabilistic output.

\section{Methods} \label{sec: methods}

This section contains our novel methodological development, which begins with specifying a Bayesian regression model for the integrand.

\subsection{A Bayesian Regression Model} \label{subsec: GP section}

This section serves to set up a generic Bayesian regression framework, which is essentially identical to that described in~\cite{OHagan1978}. See also~\cite[Section 2.7]{Rasmussen2006} and~\cite{MosamamKent2010}.
This will act as the stochastic model $f : D \times \Omega \rightarrow \mathbb{R}$ for our subsequent development.

\subsubsection{Gaussian Process Prior} \label{subsubsec: GP prior}

Recall that a \emph{Gaussian process} (GP) is a function-valued random variable $\omega \mapsto f(\cdot , \omega)$ such that \sloppy{${f(\cdot , \omega) \in C^0(D)}$} and $\omega \mapsto Lf(\cdot , \omega)$ is a (univariate) Gaussian for all continuous linear functionals $L$ on $C^0(D)$.
Here $\omega$ denotes a generic element of an underlying probability space $\Omega$.
See \cite{Bogachev1998} for further background.
Following the notational convention in \cite{Rasmussen2006}, we suppress the argument $\omega$ and denote by $f(x) \sim \mathcal{GP}(s(x),k(x,x'))$ a Gaussian process with mean function $s \in C^0(D)$ and positive definite covariance kernel $k \in C^0(D \times D)$.
The characterising property of this Gaussian process is that $f(x_1),\ldots,f(x_n)$ are jointly Gaussian with the mean vector $[s(x_1),\ldots,s(x_n)]$ and covariance matrix $[K]_{ij} = k(x_i,x_j)$ for all sets $X = \{x_1,\dots,x_n\} \subset D$.

Our starting point in this paper will be to endow a hypothetical Bayesian agent with the following prior model for the integrand:
\begin{definition}[Prior] \label{def:Prior}
Let $\pi$ denote a finite-dimensional linear subspace of real-valued functions on~$D$ and $\{p_1,\dots,p_Q\}$ a basis of $\pi$, so that $Q = \dim(\pi)$.
Then, for some positive definite covariance matrix $\Sigma \in \R^{Q \times Q}$, we consider the following hierarchical \emph{prior} model:
\begin{equation*}
f(x) \mid \gamma \sim \mathcal{GP}\big( s(x) , k(x,x') \big), \hspace{0.5cm} s(x) = \sum_{j=1}^Q \gamma_j p_j(x), \hspace{0.5cm} \gamma \sim \mathcal{N}(0,\Sigma).
\end{equation*}
\end{definition}
\noindent The mean function $s \in \pi$ is parametrised by $\gamma_1,\dots,\gamma_Q \in \mathbb{R}$.
Such a prior could arise, for example, when a parametric linear regression model is assumed and a non-parametric discrepancy term added to allow for misspecification of the parametric part \cite{Kennedy2001}.
Note that a non-zero mean $\eta \in \R^Q$ could be specified for $\gamma$; this is done in the derivations contained in supplementary material.

\subsubsection{Gaussian Process Posterior} \label{subsubsec: GP posterior}

In a regression context, the data consist of input-output pairs $\mathcal{D}_X = \{(x_i,f^\dagger(x_i))\}_{i=1}^n$, based on a finite point set $X$ that, in this paper, is considered fixed.
Our interest is in the Bayesian agent's posterior distribution, after the data $\mathcal{D}_X$ are observed.
Let $f_X$ (resp.\ $f_X^\dagger$) denote the column vector with entries $f(x_i)$ (resp.\ $f^\dagger(x_i)$).
The \emph{posterior} is defined as the law of the stochastic process which is obtained when the prior distribution is restricted to the set \sloppy{${\{\omega \in \Omega : f_X = f_X^\dagger\}}$}.
That the posterior, denoted $f \mid \mathcal{D}_X$, is again a Gaussian stochastic process is a well-known result (for technical details, see e.g.\@~\cite{Owhadi2015}).

Let $p(x)$ be the row vector with entries $p_j(x)$ and let $P_X$ denote the $n \times Q$ \emph{Vandermonde matrix} with $[P_X]_{i,j} = p_j(x_i)$.
Let $k_X(x)$ denote the row vector with entries $k(x,x_j)$ and let $K_X$ denote the \emph{kernel matrix} with $[K_X]_{i,j} = k(x_i,x_j)$.
For the prior in Def.\ \ref{def:Prior} we have the following result:
\begin{theorem}[Posterior] \label{thm:PosteriorMean}
In the posterior, $f(x) \mid \mathcal{D}_X \sim \mathcal{GP}\big(s_{X,\Sigma}(f^\dagger)(x) , k_{X,\Sigma}(x,x') \big)$ where
\begin{align}
\begin{split} \label{eqn:PosteriorMean}
s_{X,\Sigma}(f^\dagger)(x) ={}& k_X(x) \alpha + p(x) \beta  \\
={}& [k_X(x) + p(x) \Sigma P_X^\top][K_X + P_X \Sigma P_X^\top]^{-1} f_X^\dagger, 
\end{split} \\
\begin{split} \label{eqn:PosteriorCov}
k_{X,\Sigma}(x,x') ={}& k(x,x') + p(x) \Sigma p(x')^\top \\
& - [k_X(x) + p(x) \Sigma P_X^\top] [K_X + P_X \Sigma P_X^\top]^{-1} [k_X(x') + p(x') \Sigma P_X^\top]^\top 
\end{split}
\end{align}
and the coefficients $\alpha$ and $\beta$ are defined via the invertible linear system
\begin{equation}
\left[ \begin{array}{cc} K_X & P_X \\ P_X^\top & - \Sigma^{-1} \end{array} \right] \left[ \begin{array}{c} \alpha \\ \beta \end{array} \right] = \left[ \begin{array}{c} f_X^\dagger \\ 0 \end{array} \right]. \label{eqn:FirstLinearSystem} 
\end{equation}
\end{theorem}
\noindent The proofs for all results are contained in the supplementary material, unless otherwise stated.
Note that the posterior is consistent with the data $\mathcal{D}_X$, in the sense that the posterior mean $s_{X,\Sigma}(f^\dagger)(x)$ coincides with the value $f^\dagger(x)$ at the locations $x \in X$ and, moreover, the posterior variance vanishes at each $x \in X$.
These facts imply that sample paths from $f \mid \mathcal{D}_X$ almost surely satisfy $f_X = f_X^\dagger$.

\begin{remark}[Standard Bayesian cubature; BC] \label{rem: BC}
Based on Eqns.\ \ref{eqn:PosteriorMean} and \ref{eqn:PosteriorCov}, it is apparent that if we set $\pi = \emptyset$, then the posterior reduces to a Gaussian process with mean and covariance
\begin{equation*}
s_{X,0}(f^\dagger)(x) = k_X(x) K_X^{-1} f_X^\dagger, \hspace{0.5cm} k_{X,0}(x,x') = k(x,x') - k_X(x) K_X^{-1} k_X(x')^\top .
\end{equation*}
This is precisely the stochastic process used in standard Bayesian cubature~\cite{Minka2000,Briol2017}.
\end{remark}

The need for the Bayesian agent to elicit a covariance matrix $\Sigma$ appears to prevent automatic use of this prior model.
For this reason, we consider the \emph{flat prior limit} as $\Sigma^{-1} \rightarrow 0$, which corresponds to a particular encoding of an absence of prior information about the value of the parameter $\gamma$ in Def.\ \ref{def:Prior}.

\subsubsection{Flat Prior Limit}\label{sec:GPFlatPrior}

In this section we ask whether the Gaussian process posterior is well-defined in the flat prior limit $\Sigma^{-1} \to 0$.
For this, we need the concept of unisolvency \cite[Sec.\ 2.2]{Wendland2005}:

\begin{definition}[Unisolvency] \label{def:Unisolvent}
Let $\pi$ denote a finite-dimensional linear subspace of real-valued functions on $D$.
A point set $X = \{x_1,\ldots,x_n\} \subset D$ with $n \geq \dim(\polspace)$ is called $\polspace$-\emph{unisolvent} if the zero function is the only element in $\polspace$ that vanishes on $X$.
(Examples are provided in Sec. \ref{sec:unisolvency} of the supplement.)
\end{definition}

\begin{theorem}[Flat prior limit] \label{thm:LimitInterpolant}
Assume that $X$ is a $\pi$-unisolvent set.
For the prior in Def.\@~\ref{def:Prior}, and in the limit $\Sigma^{-1} \rightarrow \mathrm{0}$, we have that $s_{X,\Sigma}(f^\dagger) \rightarrow s_X(f^\dagger)$ and $k_{X,\Sigma} \rightarrow k_X$ pointwise, where
\begin{align}
s_X(f^\dagger)(x) ={}& k_X(x) \alpha + p(x) \beta, \label{eqn:FlatPriorMean}\\
\begin{split}\label{eqn:FlatPriorCov}
k_X(x,x') ={}& k(x,x') - k_X(x) K_X^{-1} k_X(x')^\top \\
& + \big[k_X(x) K_X^{-1} P_X  - p(x) \big] [P_X^\top K_X^{-1} P_X]^{-1} \big[k_X(x') K_X^{-1} P_X  - p(x') \big]^\top, 
\end{split}
\end{align}
and the coefficients $\alpha$ and $\beta$ are defined via the invertible linear system
\begin{equation}
\left[ \begin{array}{cc} K_X & P_X \\ P_X^\top & 0 \end{array} \right] \left[ \begin{array}{c} \alpha \\ \beta \end{array} \right] = \left[ \begin{array}{c} f_X^\dagger \\ 0 \end{array} \right].  \label{eqn:ThirdLinearSystem}
\end{equation}
\end{theorem}

The following observation, illustrated in Fig.\ \ref{fig:example_gp}, will be important:

\begin{proposition}\label{thm:PolyRepro}
Assume that $X$ is a $\pi$-unisolvent set. 
Then \sloppy{${s_X(p) = p}$} whenever $p \in \pi$.
\end{proposition}
\begin{proof}
If $p \in \pi$, there exist coefficients $\beta_1',\ldots,\beta_Q'$ such that $p = \sum_{i=1}^Q \beta_j' p_j$.
That is, a particular solution of Eqn.\ \ref{eqn:ThirdLinearSystem} is $\alpha = 0$ and $\beta = \beta'$.
The linear system being invertible, this must be the only solution. We deduce that $s_X(p) = p$.
\end{proof}

In particular, if $\dim(\pi) = n$, the posterior mean reduces to the unique interpolant in $\pi$ to the data $\mathcal{D}_X$ while the posterior covariance is non-zero. 
This observation will enable us to endow any cubature rule with a non-degenerate probabilistic output in Sec.\ \ref{sec:CubatureRepro}.
Next we turn our attention to estimation of the unknown value of the integral.

\begin{figure}[t]
  \centering
  \includegraphics{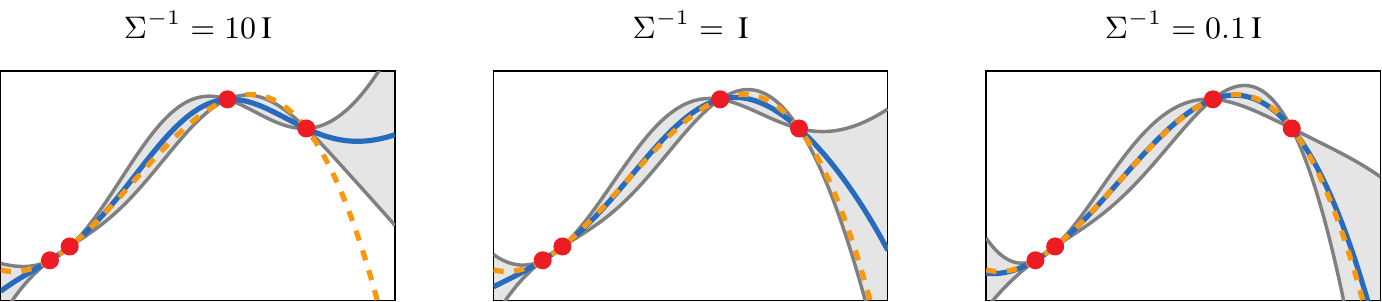}
  \caption{Posterior mean (blue) and 95\% credible intervals (gray) given four data points (red) for the prior model of Def.\@~\ref{def:Prior}, with the linear space $\pi$ taken as the set of polynomials with degree $\leq 3$. The Gaussian kernel with length-scale $\ell = 0.8$ was used. 
The unique polynomial interpolant of degree 3 to the data (dashed) is plotted for comparison.
Note convergence of the posterior mean to the polynomial interpolant as $\Sigma^{-1} \to 0$.}\label{fig:example_gp}
\end{figure}

\subsection{The Bayes--Sard Framework} \label{subsec: Bayes Sard}

Recall that the output of a generalised Bayesian cubature method is the push-forward \sloppy{${\omega \mapsto \int_D f(\cdot , \omega) \mathrm{d}\nu}$} of the stochastic process $f \mid \mathcal{D}_X$ through the integration operator $I$.
This random variable will be denoted $I(f) \mid \mathcal{D}_X$.
In this section we present the Bayes--Sard cubature method, which is based on the prior model with $\Sigma^{-1} \rightarrow 0$ studied in Sec.\ \ref{sec:GPFlatPrior}.
It will be demonstrated that Bayes--Sard cubature differs, at a fundamental level, from the standard Bayesian cubature method in that the elements of $\pi$ are exactly integrated.

Let $k_\nu(x) = I(k(\cdot,x))$ denote the \emph{kernel mean} function and $k_{\nu,\nu} = I(k_\nu)$ its integral.
Define the row vectors $p_\nu$ and $k_{\nu,X}$ to have respective entries $[p_\nu]_j = I(p_j)$ and $[k_{\nu,X}]_j = k_\nu(x_j)$.

\begin{theorem}[Bayes--Sard cubature; BSC]\label{thm:BQLimit}
Consider the Gaussian process $f \mid \mathcal{D}_X$ defined in Thm.\ \ref{thm:PosteriorMean} and suppose that $X$ is a $\pi$-unisolvent point set.
Then in the limit $\Sigma^{-1} \rightarrow 0$ we have that \sloppy{${I(f) \mid \mathcal{D}_X \sim \mathcal{N}( \mu_X(f^\dagger) , \sigma_X^2 )}$} with
\begin{equation*}
\mu_X(f^\dagger) = w_k^\top f_X^\dagger \hspace{0.5cm} \text{and} \hspace{0.5cm} \sigma_X^2 = k_{\nu,\nu} - k_{\nu,X} K_X^{-1} k_{\nu,X}^\top + \big( k_{\nu,X} K_X^{-1} P_X - p_\nu \big) w_\pi,
\end{equation*}
where the weight vectors $w_{k} \in \R^n$ and $w_{\pi} \in \R^Q$ are obtained from the solution of the linear system
\begin{equation}\label{eqn:FlatPriorWeights}
\left[ \begin{array}{cc} K_X & P_X \\ P_X^\top & 0 \end{array} \right] \left[ \begin{array}{c} w_{k} \\ w_{\pi} \end{array} \right]  =  \left[ \begin{array}{c} k_{\nu,X}^\top \\ p_\nu^\top \end{array} \right].
\end{equation}
\end{theorem}

The posterior mean indeed takes the form of a cubature rule, with weights $w_{k,i}$ and points $x_i \in X$.
This provides a point estimator for the integral $I(f^\dagger)$, while the posterior variance enables uncertainty to be assessed.
The \emph{Bayes--Sard} nomenclature derives from the fact that the associated cubature rule $\mu_X$ is \emph{exact} on the space $\pi$ (recall Prop.\ \ref{thm:PolyRepro}; the proof is also similar):

\begin{proposition}
Assume that $X$ is a $\pi$-unisolvent set. Then $\mu_X(p) = I(p)$ whenever $p \in \pi$.
\end{proposition}
Thus we have a probabilistic framework that combines the flexibility of Bayesian cubature with the robustness of classical numerical integration techniques, for instance based on a polynomial exactness criteria being satisfied.

\subsection{Normalised Bayesian Cubature}

The difference between BSC and BC is perhaps best illustrated in the case $\pi = \{1\}$, where constant functions are exactly integrated in BSC but not in BC.
Indeed, $P_X = \mathbbm{1}$, the $n$-vector of ones, and
\begin{equation*}
w_k = \bigg( \text{I} - \frac{K_X^{-1} \mathbbm{1} \mathbbm{1}^\top }{\mathbbm{1}^\top K_X^{-1} \mathbbm{1}} \bigg) K_X^{-1} k_{\nu,X}^\top + \frac{K_X^{-1} \mathbbm{1} }{\mathbbm{1}^\top K_X^{-1} \mathbbm{1}}.
\end{equation*}
These weights have the desirable property of summing up to one; we might therefore call this a \emph{normalised Bayesian cubature} method.
Furthermore, if the kernel is parametrised by a length-scale parameter and this parameter is too small, then $w_{k,i} \approx 1/n$, which is a reasonable default.
This should be contrasted with BC, for which the weights $w_{k,i} \approx 0$ become degenerate instead.

\subsection{Reproduction of Classical Cubature Rules}\label{sec:CubatureRepro}

In this section we indicate how \emph{any} cubature rule can be endowed with a probabilistic interpretation under the Bayes--Sard framework.
Recall that every continuous positive definite kernel $k$ induces a unique \emph{reproducing kernel Hilbert space} (RKHS) $H(k) \subset C^0(D)$ with norm denoted $\norm{\cdot}_k$~\cite{Berlinet2011}.
It is well-known that the weights $w_\text{BC} := K_X^{-1} k_{\nu,X}^\top \in \R^n$ of the standard Bayesian cubature method (recall Rmk. \ref{rem: BC}) are \emph{worst-case optimal} in $H(k)$:
\begin{equation*}
w_\text{BC} = \argmin_{w \in \R^n} e_k(X,w), \qquad e_k(X,w) := \sup_{ \norm{h}_k \leq 1 } \, \abs[\bigg]{\int_D h \mathrm{d} \nu - \sum_{i=1}^n w_i h(x_i) },
\end{equation*}
where $e_k(X,w)$ is the \emph{worst-case error} (WCE) of the cubature rule specified by the points $X$ and weights $w$.
Furthermore, the posterior standard deviation coincides with $e_k(X,w_\text{BC})$.
See~\cite{Larkin1970,RichterDyn1971b,Oettershagen2017} for further details on optimal cubature rules in RKHS.
It is now shown that, when $\dim(\pi) = n$, the Bayes--Sard weights in Thm.\ \ref{thm:BQLimit} do not depend on the kernel and the standard deviation coincides with the WCE:

\begin{theorem}\label{thm:PolyReproBC}
Suppose that $\dim(\pi) = n$ and let $X$ be a $\pi$-unisolvent set. Then
\begin{equation*}
\mu_X(f^\dagger) = w_k^\top f_X^\dagger, \hspace{0.5cm}w_k^\top = p_\nu P_X^{-1}, \hspace{0.5cm} \mu_X(p) = I(p) \hspace{0.5cm} \text{for every $p \in \pi$} \hspace{0.5cm} 
\end{equation*}
and
\begin{equation*}
\sigma_X^2 = e_k(X,w_k)^2 = k_{\nu,\nu} - 2 k_{X,\nu} w_k + w_k^\top K_X w_k.
\end{equation*}
That is, the Bayes--Sard cubature weights $w_k$ are the unique weights such that every function in $\pi$ is integrated exactly and the posterior standard deviation $\sigma_X$ coincides with the WCE in the RKHS $H(k)$.
\end{theorem}

\begin{corollary} \label{cor:cubatureRepro}
Consider an $n$-point cubature rule with points $X$ and non-zero weights $w \in \R^n$ and assume that $\nu$ admits a positive density function.
Then there is a function space $\pi$ of dimension $n$, such that the Bayes--Sard method recovers $w_k = w$ and $\sigma_X^2 = e_k(X, w)^2$, as defined in Thm.\ \ref{thm:BQLimit}.
\end{corollary}

Thus \emph{any} cubature rule can be recovered as a posterior mean for some prior.
Our result goes beyond earlier work in \cite{Sarkka2016,Karvonen2017a}, in the sense that the associated posterior is non-degenerate (i.e.\ has non-zero variance) in the Bayes--Sard framework.
Further discussion is provided in Sec. \ref{sec:kernel-perspective} of the supplement.
From a practical perspective, this enables us to simultaneously achieve the same reliable integration performance as popular non-probabilistic rules (see Sec.\ \ref{subsec:Q=n-rates} in the supplement) \emph{and} to perform formal uncertainty quantification for the integral.

\subsection{Convergence Results}\label{sec:theory}

This section contains fundamental convergence results for the cubature rule $\mu_X$ associated with the mean of the Bayes--Sard output.
For standard BC, the analogous convergence results can be found in \cite{Briol2017,Kanagawa2017}.
Our attention is restricted to the case when $\pi$ is the space $\Pi_{m}(\R^d)$ of $d$-variate polynomials of degree at most $m \geq 0$:
\begin{equation*}
\Pi_m(\R^d) \coloneqq \mathrm{span} \{ x^\alpha \colon \abs{\alpha} \leq m \}
\end{equation*}
where $\alpha \in \N_0^d$ is a multi-index, $x^\alpha = x_1^{\alpha_1} \cdots x_d^{\alpha_d}$ and $\abs{\alpha} = \alpha_1 + \cdots + \alpha_d$.
It is noteworthy that Thm.\ \ref{thm:maternConvergence} has been derived, essentially in the form we present it, in non-probabilistic setting already in~\cite{Bezhaev1991}.
However, we go beyond \cite{Bezhaev1991} and provide convergence results for both the Gaussian kernel, as well as kernels of the Mat\'{e}rn class that are used in the numerical experiments in Sec.\ \ref{sec:numerics}.

To establish convergence, we observe that the posterior mean $s_X(f^\dagger)$ defined in Eqn.\ \ref{eqn:FlatPriorMean} coincides with the interpolant defined in \cite[Section 8.5]{Wendland2005} for a conditionally positive definite kernel\footnote{Note that a positive definite kernel is also a conditionally positive definite kernel.}.
The extensive convergence theory outlined in \cite[Chapter 11]{Wendland2005} can be therefore brought to bear.
For a set $X \subset D$ and $D$ bounded, define the \emph{fill distance} $h_{X,D} \coloneqq \sup_{x \in D} \, \min_{i = 1,\dots,n} \, \|x - x_i\|$.
Considered as a sequence of sets indexed by $n \in \mathbb{N}$, we say $X$ is \emph{quasi-uniform} in $D$ if $h_{X,D} \lesssim n^{-1/d}$, where $a_n \lesssim b_n$ is used to signify that the ratio $a_n / b_n$ is bounded above for sufficiently large $n \in \mathbb{N}$.

\begin{theorem}[Spectral convergence for Gaussian kernels] \label{thm:conv-Gaussian}
Let $D$ be a hypercube in $\R^d$, let $\nu$ admit a density which is bounded, let $X$ be a $\Pi_{m}(\R^d)$-unisolvent set for some $m \geq 0$, and let $k$ be a Gaussian kernel: $k(x,x') = \exp(-\norm{x-x'}^2/(2\ell^2))$ for some $\ell > 0$.
Then there is a $c > 0$ such that, for a quasi-uniform point set and any $\varepsilon > 0$,
\begin{equation*}
\abs{ \mu_X(f^\dagger) - I(f^\dagger) } \lesssim \mathrm{e}^{-c n^{1/d - \varepsilon}} \|f^\dagger\|_k.
\end{equation*}
\end{theorem}
\begin{proof}
That for any $\varepsilon > 0$ there is a $h_0 > 0$ such that $\abs{ \mu_X(f^\dagger) - I(f^\dagger) } \lesssim \mathrm{e}^{-c/h_{X,D}^{1-\varepsilon}} \|f^\dagger\|_k$ whenever $h_{X,D} < h_0$ is established by \cite[Thm.\ 11.22]{Wendland2005}.
The remainder of the proof is transparent.
\end{proof}

The next result extends~\cite[Prop.\ 4]{Kanagawa2017} for the standard Bayesian cubature method. 
Its proof follows that of Thm.~\ref{thm:conv-Gaussian} and is an application of~\cite[Cor.\ 11.33]{Wendland2005}.

\begin{theorem}[Polynomial convergence for Sobolev kernels]\label{thm:maternConvergence}
Let $X$ be a $\Pi_{m}(\R^d)$-unisolvent set for some $m \geq 0$. Suppose that (i) $D$ is a bounded open set that satisfies an interior cone condition and whose boundary is Lipschitz; (ii) for $\alpha > d/2$, the RKHS of the kernel $k$ is norm-equivalent to the standard Sobolev space $H^\alpha(D)$ and (iii) $\nu$ admits a density function that is bounded.
Then, for a quasi-uniform point set,
\begin{equation*}
\abs{ \mu_X(f^\dagger) - I(f^\dagger) } \lesssim n^{-\alpha/d} \|f^\dagger\|_{H^\alpha(D)}.
\end{equation*}
\end{theorem}

\begin{remark}
The classical Mat\'{e}rn kernel with smoothness parameter $\rho$ satisfies Assumption (ii) of the above theorem with $\alpha = \rho + d/2$.
Sec. \ref{subsec:maternDetails} of the supplementary material elaborates further on Assumptions (i) and (ii).
\end{remark}

This completes our theoretical convergence analysis for the Bayes--Sard method.

\section{Experimental Results} \label{sec:numerics}

This section contains two numerical experiments, which investigate the empirical performance of the Bayes--Sard cubature method and the associated uncertainty quantification that is provided.
The examples demonstrate that Bayes--Sard cubature is typically at least as accurate as standard Bayesian cubature whilst being less sensitive to misspecification of the kernel length-scale parameter.

\subsection{On Choosing the Kernel Parameters} \label{sec:kernel-parameters}

The stationary kernels typically used in Gaussian process regression are parametrised by positive \emph{length-scale}\footnote{In general, a distinct length scale parameter for each dimension could be used.} and \emph{amplitude} parameters $\ell$ and $\lambda$, respectively:
\begin{equation*}
k(x,x') = k(x-x') = \lambda k_0 \big( (x-x')/\ell \big)
\end{equation*}
for, in a slight abuse of notation, some base kernel $k_0$.
Adapting these parameters in a data-dependent way is an essential prerequisite for meaningful quantification of uncertainty for the integral.
In the first example we set these parameters independently, following the approach suggested in~\mbox{\cite[Sec.\ 4.1]{Briol2017}}.
We (i) assign $\lambda$ an improper prior and marginalise over it so that the BSC posterior becomes Student-$t$ with the mean $\mu_X(f^\dagger)$, variance $((f_X^\dagger)^\top K_{X}^{-1} f_X^\dagger/n)\,\sigma_X^2$ and $n$ degrees of freedom, and (ii) set $\ell$ using empirical Bayes (EB) based on the standard GP log-marginal likelihood~\cite[Sec.\ 5.4.1]{Rasmussen2006}.
There are of course other possibilities that could be explored, such as cross-validation or using the likelihood of the regression model set up in Sec.\ \ref{subsec: GP section} (see \cite[Eqn.\ 2.45]{Rasmussen2006}).

\begin{figure}[t]
  \centering
  \includegraphics{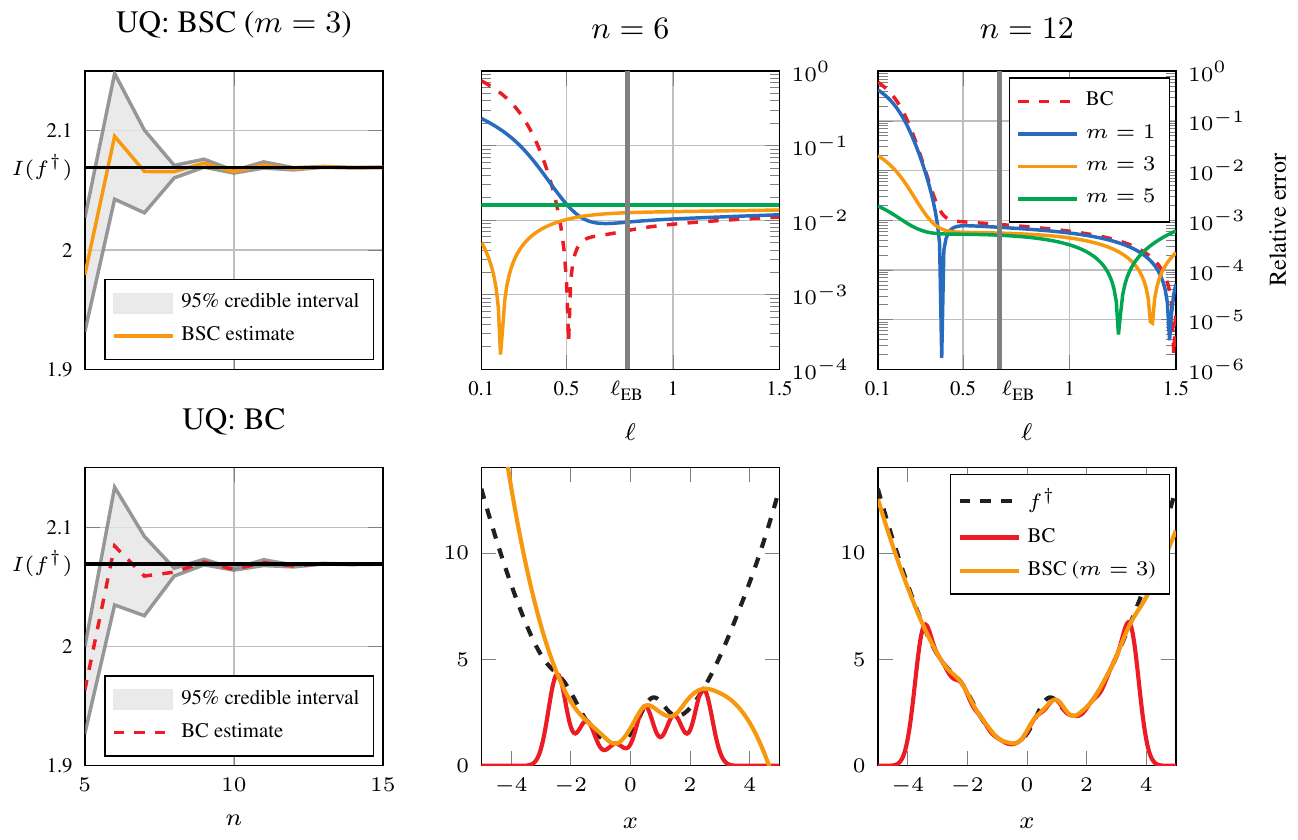}
  \caption{
  Approximation of the integral in Eqn.\ \ref{eqn:ToyIntegral} using Bayesian cubature (BC) and Bayes--Sard cubature (BSC) with $\pi = \Pi_m(\R)$ for $m=1,3,5$, both based on the Gaussian kernel.
  The $n$ nodes were placed uniformly on $[-\sqrt{n},\sqrt{n} \, ]$.
  \emph{Left}: Uncertainty quantification (UQ) provided by BSC with $m = 3$ and BC when kernel parameters were selected as outlined in Sec.\ \ref{sec:kernel-parameters}.
  \emph{Middle} \& \emph{right}: Effect of the length-scale on approximation accuracy.
  The upper row presents the relative integration error $\abs{I(f^\dagger) - I_n(f^\dagger)}/I(f^\dagger)$ for each cubature rule $I_n$ as a function of the length-scale.
  The  ``optimal'' length-scales $\ell_\text{EB}$, as computed by empirical Bayes, are also displayed.
  The lower row contains the corresponding posterior means when an inappropriate value, $\ell = 0.3$, is used. 
  Since $\dim(\Pi_5(\R)) = 6$, that BSC for $m=5$ and $n=6$ is independent of $\ell$ is a consequence of Thm.\ \ref{thm:PolyReproBC}.
  }\label{fig:1D-quad}
\end{figure}

\subsection{A One-Dimensional Toy Example}\label{sec:toy-example}

Our first example is one-dimensional.
The test function and its integral that we considered were
\begin{equation}\label{eqn:ToyIntegral}
f^\dagger(x) = \exp\left( \sin(2x) - \frac{x^2}{5} \right) + \frac{x^2}{2} \hspace{0.5cm} \text{and} \hspace{0.5cm} I(f^\dagger) = \frac{1}{\sqrt{2\pi}} \int_\R f^\dagger(x) \mathrm{e}^{-x^2/2} \mathrm{d} x \approx 2.0693.
\end{equation}
The effect of the length-scale $\ell$ of the Gaussian kernel on the performance of standard Bayesian cubature and Bayes--Sard cubature of Sec.\ \ref{subsec: Bayes Sard}, with $\pi = \Pi_m(\R)$ for different $m$, was investigated and the quality of the uncertainty quantification was assessed.

Results are depicted in Fig.\ \ref{fig:1D-quad}.
It can be observed that the BSC is more robust compared to Bayesian cubature when the length-scale is misspecified (in particular, when it is too small).
This is because the polynomial part mitigates the tendency of the posterior mean to revert quickly back to zero.
For reasonable values of the length-scale, the accuracy of the different methods is comparable.
The BSC and BC provide qualitatively similar quantification of uncertainty and both exhibit a degree of over-confidence, as observed already in~\cite[Sec.\ 5.1]{Briol2017} for the BC and attributed to the manner in which the length scale is selected.
However, BSC is less over-confident.
The reason for this is that the BSC variance in Thm.\ \ref{thm:BQLimit} is a sum of the BC variance and a positive term.

\subsection{Zero Coupon Bonds}\label{sec:zcb}

This section experiments with a high-dimensional zero coupon bond example that has been used previously in numerical experiments for kernel cubature in~\cite[Sec.\ 5.5]{Karvonen2018}.
See~\cite[Sec.\ 6.1]{Holtz2011} and Sec.\@~\ref{sec:experiments} of the supplement for a more detailed account of this experiment.
The integral of interest is
\begin{equation}\label{eq:zcbInt}
P(0,T) \coloneqq \mathbb{E}\Bigg[ \exp\bigg( -\Delta t \sum_{i=0}^{D-1} r_{t_i} \bigg)\Bigg] = \exp(-\Delta t r_{t_0}) \mathbb{E}\Bigg[ \exp\bigg( -\Delta t \sum_{i=1}^{D-1} r_{t_i} \bigg)\Bigg],
\end{equation}
where $r_{t_i}$ are Gaussian random variables. This $d = D-1$ dimensional integral represents the price at time $t = 0$ of a zero coupon bond with maturity time $T$ and arises from $D$-step uniform Euler--Maruyama discretisation of the Vasicek model. Existence of a closed-form solution for $P(0,T)$ makes numerical approximation of Eqn.\ \ref{eq:zcbInt} an attractive high-dimensional benchmark problem. 

We transformed the integral~\eqref{eq:zcbInt} onto the hypercube $[0,1]^d$ and compared the accuracy of BC to BSC with $\pi = \Pi_1(\R^d)$.
Different dimensions $d$ and length-scales $\ell$ were considered and the product Matérn kernel with smoothness parameter $\rho = 5/2$ was used.
As in Sec.\ \ref{sec:toy-example}, it is apparent from Fig.\@~\ref{fig:zcb} that the BSC is less sensitive to misspecification of the length-scale parameter compared to the standard BC method.
In this misspecified case a two order of magnitude reduction in integration error was observed.
This is attributed to the improved extrapolation performance conferred through the polynomial component.

\begin{figure}[t]
  \centering
  \includegraphics[width=\textwidth]{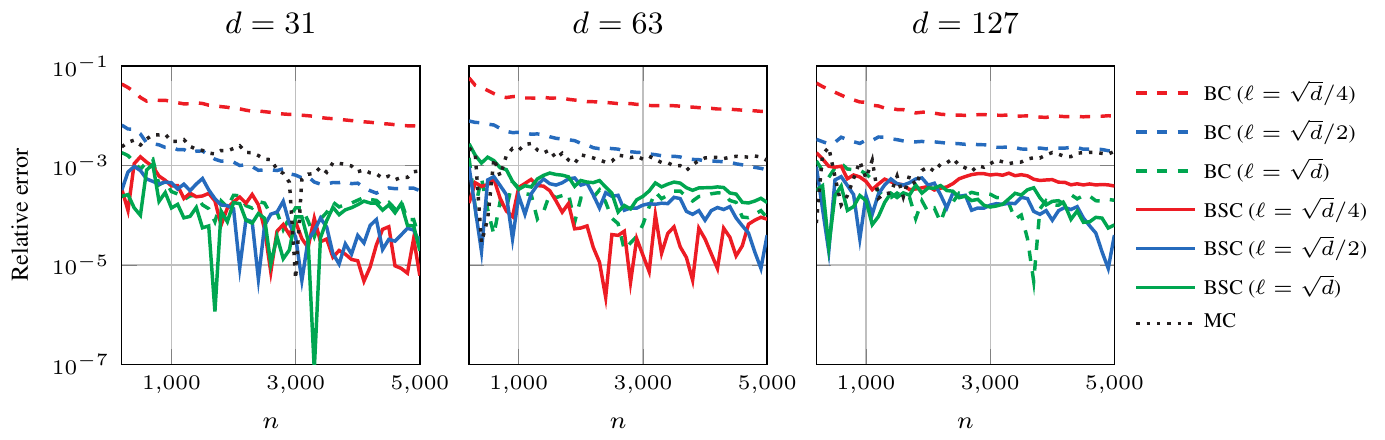}
  \caption{
  Approximation of the $d$-dimensional integral~\eqref{eq:zcbInt} using Bayesian cubature (BC) and Bayes--Sard cubature (BSC) with $\pi = \Pi_1(\R^d)$, both based on product Matérn kernel with $\rho = 5/2$ and length-scale $\ell$. 
  Figures contain relative integration errors for each cubature rule for a given dimension and different length-scales as a function of the number of nodes $n$, drawn randomly from the uniform distribution. 
  Probability of the point set being not unisolvent is zero as this would require the points to lie on a hyperplane.
  The standard Monte Carlo approximation (MC) is plotted for comparison.}\label{fig:zcb}
\end{figure}

\section{Conclusion}\label{sec:discussion}

This paper proposed a Bayes--Sard cubature method, which provides an explicit connection between classical cubatures and the Bayesian inferential framework.
In particular, we obtained polynomially exact generalisations of standard Bayesian cubature in Thm.\ \ref{thm:BQLimit} and demonstrated in Cor.\@~\ref{cor:cubatureRepro} how any cubature rule, including widely-used cubature methods, can be recovered as the output of a probabilistic model.

The main practical drawback of standard Bayesian cubature is its acute sensitivity to the choice of kernel. 
As demonstrated in Sec.\ \ref{sec:numerics}, the Bayes--Sard point estimator performance is more robust to the choice of kernel and this suggests that fast GP methods (e.g.,\@~\cite{Hensman2018,ZhouEtAl2018}) could be used for efficient automatic selection of kernel parameters with little adverse effect on accuracy of the point estimator.
On the other hand, further work is required to assess the quality of the uncertainty quantification provided by the Bayes--Sard method.
This will require careful analysis that accounts for how kernel parameters are estimated, and is expected to be technically more challenging (see, e.g.,~\cite{Xu2017}).

\subsubsection*{Acknowledgments}

The authors are grateful for discussion with Aretha Teckentrup, Catherine Powell and Fred Hickernell.
TK was supported by the Aalto ELEC Doctoral School.
CJO was supported by the Lloyd's Register Foundation programme on data-centric engineering. 
SS was supported by the Academy of Finland projects 266940 and 304087.
This material was based upon work partially supported by the National Science Foundation under Grant DMS-1127914 to the Statistical and Applied Mathematical Sciences Institute. Any opinions, findings, and conclusions or recommendations expressed in this material are those of the author(s) and do not necessarily reflect the views of the National Science Foundation.

\newpage



\newtheorem{innercustomthm}{Theorem}
\newenvironment{customtheorem}[1]
  {\renewcommand\theinnercustomthm{#1}\innercustomthm}
  {\endinnercustomthm}

\newtheorem{innercustomcor}{Corollary}
\newenvironment{customcorollary}[1]
  {\renewcommand\theinnercustomcor{#1}\innercustomcor}
  {\endinnercustomcor}  

\newtheorem{innercustomdef}{Definition}
\newenvironment{customdefinition}[1]
  {\renewcommand\theinnercustomdef{#1}\innercustomdef}
  {\endinnercustomdef}  

\setcounter{equation}{0}
\renewcommand{\theequation}{\thesection\arabic{equation}}


\appendix

\begin{center}
{\huge \textbf{Supplementary Material}}
\end{center}

\vspace{30pt}


\section{Proof of Results in the Main Text}

The prior model (Def.\ \ref{def:Prior}) used in the main text is
\begin{equation*}
f(x) \mid \gamma \sim \mathcal{GP}\big( s(x) , k(x,x') \big), \hspace{0.5cm} s(x) = \sum_{j=1}^Q \gamma_j p_j(x), \hspace{0.5cm} \gamma \sim \mathcal{N}(0,\Sigma).
\end{equation*}
It is straightforward to consider the generalisation $\gamma \sim \mathcal{N}(\eta, \Sigma)$ for potentially non-zero vector $\eta \in \R^Q$; we do this in this supplement.

\subsection{Results on the Regression Model}

\begin{customtheorem}{\ref{thm:PosteriorMean}}[Posterior]
In the posterior, $f(x) \mid \mathcal{D}_X \sim \mathcal{GP}\big(s_{X,\Sigma}(f^\dagger)(x) , k_{X,\Sigma}(x,x') \big)$ where
\begin{align}
\begin{split} \label{supplement-eqn:PosteriorMean}
s_{X,\Sigma}(f^\dagger)(x) ={}& k_X(x) \alpha + p(x) \beta  \\
={}& [k_X(x) + p(x) \Sigma P_X^\top][K_X + P_X \Sigma P_X^\top]^{-1} f_X^\dagger, 
\end{split} \\
\begin{split} \label{supplement-eqn:PosteriorCov}
k_{X,\Sigma}(x,x') ={}& k(x,x') + p(x) \Sigma p(x')^\top \\
& - [k_X(x) + p(x) \Sigma P_X^\top] [K_X + P_X \Sigma P_X^\top]^{-1} [k_X(x') + p(x') \Sigma P_X^\top]^\top 
\end{split}
\end{align}
and the coefficients $\alpha$ and $\beta$ are defined via the invertible linear system
\begin{equation}
\left[ \begin{array}{cc} K_X & P_X \\ P_X^\top & - \Sigma^{-1} \end{array} \right] \left[ \begin{array}{c} \alpha \\ \beta \end{array} \right] = \left[ \begin{array}{c} f_X^\dagger \\ -\eta \end{array} \right]. \label{supplement-eqn:FirstLinearSystem} 
\end{equation}
\end{customtheorem}

\begin{proof}
Under the hierarchical prior we have the marginal
\begin{equation*}
f(x) \sim \mathcal{GP} \big( p(x) \eta , k(x,x') + p(x) \Sigma p(x')^\top \big).
\end{equation*}
Thus standard formulae for the conditioning of a Gaussian process \cite[Eqns.\ 2.25, 2.26]{Rasmussen2006} can be used:
\begin{align}
s_{X,\Sigma}(f^\dagger)(x) ={}& p(x) \eta + [k_X(x) + p(x) \Sigma P_X^\top] [K_X + P_X \Sigma P_X^\top]^{-1} [f_X^\dagger - P_X \eta], \nonumber \\
\begin{split} \label{supplement-eq:PosteriorVarSigma}
k_{X,\Sigma}(x,x') ={}& k(x,x') + p(x) \Sigma p(x')^\top \\
& - [k_X(x) + p(x) \Sigma P_X^\top] [K_X + P_X \Sigma P_X^\top]^{-1} [k_X(x') + p(x') \Sigma P_X^\top]^\top .
\end{split}
\end{align}
The coefficients $\alpha$ and $\beta$ are therefore
\begin{align*}
\alpha & = [K_X + P_X \Sigma P_X^\top]^{-1} [f_X^\dagger - P_X \eta], \\
\beta & = \eta +  \Sigma P_X^\top [K_X + P_X \Sigma P_X^\top]^{-1} [f_X^\dagger - P_X \eta].
\end{align*}
It can be verified by substitution that $P_X^\top \alpha - \Sigma^{-1} \beta = -\eta$ and that the interpolation equations \sloppy{${K_X \alpha + P_X \beta = f_X^\dagger}$} hold.
This allows us to provide the equivalent characterisation of $\alpha$ and $\beta$ in terms of the linear system in Eqn.\ \ref{supplement-eqn:FirstLinearSystem}.
To see that this linear system is invertible, we can use the block matrix determinant formula
\begin{equation*}
\det \bigg( \left[ \begin{array}{cc} K_X & P_X \\ P_X^\top & - \Sigma^{-1} \end{array} \right] \bigg) = \det(-\Sigma^{-1}) \det(K_X + P_X \Sigma P_X^\top).
\end{equation*}
That is, since $\Sigma$ is a positive definite covariance matrix, the block matrix is invertible if and only if $K_X + P_X \Sigma P_X^\top$ is invertible. 
This is indeed true because, for instance, $K_X + P_X \Sigma P_X^\top$ is the covariance matrix for the random vector $f_X$ under the prior, which is non-singular.
\end{proof}

The following Lagrange form~\cite[Section 11.1]{Wendland2005} of the posterior will be useful:

\begin{theorem}[Lagrange form for the posterior]\label{supplement-thm:LagrangeForm}
The posterior mean and covariance functions in Eqns.\ \ref{supplement-eqn:PosteriorMean} and \ref{supplement-eqn:PosteriorCov} can be written in the \emph{Lagrange form}
\begin{align}
s_{X,\Sigma}(f^\dagger)(x) &= u_{X,\Sigma}(x)^\top f_X^\dagger - v_{X,\Sigma}(x)^\top \eta, \nonumber \\
k_{X,\Sigma}(x,x') &= k(x,x') + p(x) \Sigma p(x')^\top - [k_X(x) + p(x) \Sigma P_X^\top] u_{X,\Sigma}(x'), \label{supplement-eq:CovLagrange}
\end{align}
where 
\begin{equation}\label{supplement-eqn:LagrangeCardFuncs}
u_{X,\Sigma}(x) \coloneqq [K_X + P_X \Sigma P_X^\top]^{-1}[k_X(x) + p(x) \Sigma P_X^\top]^\top
\end{equation}
is a vector of \emph{Lagrange cardinal functions} and $v_{X,\Sigma}(x) \coloneqq \Sigma [ P_X^\top u_{X,\Sigma}(x) - p(x)^\top ]$.
These functions are obtained from the invertible linear system
\begin{equation}\label{supplement-eqn:LagrangeSystem}
\left[ \begin{array}{cc} K_X & P_X \\ P_X^\top & - \Sigma^{-1} \end{array} \right] \left[ \begin{array}{c} u_{X,\Sigma}(x) \\ v_{X,\Sigma}(x) \end{array} \right] = \left[ \begin{array}{c} k_X(x)^\top \\ p(x)^\top \end{array} \right]
\end{equation}
and satisfy the \emph{cardinality property} $[u_{X,\Sigma}(x_j)]_i = \delta_{ij}$ and $[v_{X,\Sigma}(x_j)]_i = 0$ for every $i,j \in \{ 1,\ldots,n \}$.
\end{theorem}
\begin{proof}
From Eqns.\ \ref{supplement-eqn:PosteriorMean} and \ref{supplement-eqn:FirstLinearSystem}, the posterior mean is
\begin{equation*}
s_{X,\Sigma}(f^\dagger)(x) = \left[ \begin{array}{cc} k_X(x) & p(x) \end{array} \right] \left[ \begin{array}{c} \alpha \\ \beta \end{array} \right] = \left[ \begin{array}{cc} k_X(x) & p(x) \end{array} \right] \left[ \begin{array}{cc} K_X & P_X \\ P_X^\top & - \Sigma^{-1} \end{array} \right]^{-1}  \left[ \begin{array}{c} f_X^\dagger \\ -\eta \end{array} \right],
\end{equation*}
and this can be written as $s_{X,\Sigma}(f^\dagger)(x) = u_{X,\Sigma}(x)^\top f_X^\dagger - v_{X,\Sigma}(x)^\top \eta$ where $u_{X,\Sigma}(x)$ and $v_{X,\Sigma}(x)$ are obtained from the linear system in Eqn.\ \ref{supplement-eqn:LagrangeSystem}.
The expression for the posterior covariance follows by inserting $u_{X,\Sigma}(x')$, as given in Eqn.\ \ref{supplement-eqn:LagrangeCardFuncs}, into Eqn.\ \ref{supplement-eqn:PosteriorCov}.
The cardinality property follows after we recognise that, setting $x = x_j$, Eqn.\ \ref{supplement-eqn:LagrangeSystem} is solved by $u_{X,\Sigma}(x_j) = e_j$ (the $j$th unit coordinate vector) and $v_{X,\Sigma}(x_j) = 0$.
\end{proof}

\begin{customtheorem}{\ref{thm:LimitInterpolant}}[Flat prior limit]
Assume that $X$ is a $\pi$-unisolvent set.
For the prior in Def.\@~\ref{def:Prior}, and in the limit $\Sigma^{-1} \rightarrow \mathrm{0}$, we have that $s_{X,\Sigma}(f^\dagger) \rightarrow s_X(f^\dagger)$ and $k_{X,\Sigma} \rightarrow k_X$ pointwise, where
\begin{align}
s_X(f^\dagger)(x) ={}& k_X(x) \alpha + p(x) \beta, \label{supplement-eqn:FlatPriorMean}\\
\begin{split}\label{supplement-eqn:FlatPriorCov}
k_X(x,x') ={}& k(x,x') - k_X(x) K_X^{-1} k_X(x')^\top \\
& + \big[k_X(x) K_X^{-1} P_X  - p(x) \big] [P_X^\top K_X^{-1} P_X]^{-1} \big[k_X(x') K_X^{-1} P_X  - p(x') \big]^\top, 
\end{split}
\end{align}
and the coefficients $\alpha$ and $\beta$ are defined via the invertible linear system
\begin{equation}
\left[ \begin{array}{cc} K_X & P_X \\ P_X^\top & 0 \end{array} \right] \left[ \begin{array}{c} \alpha \\ \beta \end{array} \right] = \left[ \begin{array}{c} f_X^\dagger \\ -\eta \end{array} \right].  \label{supplement-eqn:ThirdLinearSystem}
\end{equation}
\end{customtheorem}
\begin{proof}
For the mean function, the limit can just be taken in the linear system of Eqn.\ \ref{supplement-eqn:FirstLinearSystem} and it is required is to verify that this system can be inverted.
From an application of the formula for a block matrix determinant we have that the determinant of the matrix in Eqn.\@~\ref{supplement-eqn:ThirdLinearSystem} equals $\text{det}(- P_X^\top K_X^{-1} P_X) \text{det}(K_X)$, where $\text{det}(K_X) > 0$.
Because $X$ is $\pi$-unisolvent, $P_X$ is of full rank and consequently the matrix $P_X^\top K_X^{-1} P_X$ is invertible.
Thus the block matrix is invertible.

To obtain the covariance function an additional argument is needed. 
To this end, the Woodbury matrix identity yields
\begin{equation*}
[K_X + P_X \Sigma P_X^\top]^{-1} = K_X^{-1} - K_X^{-1} P_X [ \Sigma^{-1} + P_X^\top K_X^{-1} P_X ]^{-1} P_X^\top K_X^{-1}.
\end{equation*}
Denoting $L_X \coloneqq P_X^\top K_X^{-1} P_X$ and inserting the above into Eqn.\@~\ref{supplement-eq:PosteriorVarSigma} produces
\begin{equation}\label{supplement-eq:PostVarSigmaExpanded}
\begin{split}
k_{X,\Sigma}(x,x') ={}& k(x,x') - k_X(x) K_X^{-1} k_X(x')^\top + p(x) \Sigma p(x')^\top - p(x) \Sigma L_X \Sigma p(x')^\top \\
& - k_X(x) K_X^{-1} P_X \Sigma p(x')^\top - p(x) \Sigma P_X^\top K_X^{-1} k_X(x')^\top \\
&+ k_X(x) K_X^{-1}P_X [ \Sigma^{-1} + L_X ]^{-1} P_X^\top K_X^{-1} k_X(x')^\top \\
&+ k_X(x) K_X^{-1}P_X [ \Sigma^{-1} + L_X ]^{-1} L_X \Sigma p(x')^\top \\
&+ p(x) \Sigma L_X [ \Sigma^{-1} + L_X]^{-1} P_X^\top K_X^{-1} k_X(x')^\top \\
&+ p(x) \Sigma L_X [ \Sigma^{-1} + L_X ]^{-1} L_X \Sigma p(x')^\top.
\end{split}
\end{equation}
For small enough $\Sigma^{-1}$ we can write the Neumann series
\begin{equation*}
[ \Sigma^{-1} + L_X ]^{-1} = L_X^{-1} \big[ I - (L_X \Sigma)^{-1} + (L_X \Sigma)^{-2} - \cdots \big].
\end{equation*}
Therefore we have the trio of results
\begin{align*}
K_X^{-1}P_X [ \Sigma^{-1} + L_X ]^{-1} P_X^\top K_X^{-1} &= K_X^{-1}P_X L_X^{-1} P_X^\top K_X^{-1} + \mathcal{O}(\Sigma^{-1}),  \\
K_X^{-1}P_X [ \Sigma^{-1} + L_X ]^{-1} L_X \Sigma &= K_X^{-1} P_X \Sigma - K_X^{-1} P_X L_X^{-1} + \mathcal{O}(\Sigma^{-1}), \\
\Sigma L_X [ \Sigma^{-1} + L_X ]^{-1} L_X \Sigma &= \Sigma L_X \Sigma - \Sigma + L_X^{-1} + \mathcal{O}(\Sigma^{-1}).
\end{align*}
Inserting these into Eqn.\@~\ref{supplement-eq:PostVarSigmaExpanded} yields, after cancellation and taking the limit $\Sigma^{-1} \to 0$,
\begin{equation*}
\begin{split}
k_{X}(x,x') ={}& k(x,x') - k_X(x) K_X^{-1} k_X(x')^\top \\
&+ \big[k_X(x) K_X^{-1} P_X  - p(x) \big] [P_X^\top K_X^{-1} P_X]^{-1} \big[k_X(x') K_X^{-1} P_X  - p(x') \big]^\top,
\end{split}
\end{equation*}
as claimed.
\end{proof}

The flat prior limit of the posterior also admits a Lagrange representation:

\begin{theorem}[Lagrange form in the flat prior limit] \label{supplement-thm:FlatPriorLagrange}
Assume that $X$ is a $\pi$-unisolvent set.
The posterior mean and covariance in Eqns.\ \ref{supplement-eqn:FlatPriorMean} and \ref{supplement-eqn:FlatPriorCov} can be written as
\begin{align}
s_X(f^\dagger)(x) &= u_X(x) f_X^\dagger - v_X(x) \eta, \nonumber \\
k_X(x,x') &= k(x,x') - k_X(x) K_X^{-1} k_X(x')^\top + \big[k_X(x) K_X^{-1} P_X - p(x) \big] v_X(x') \label{supplement-eqn:FlatPriorCovLagrange}
\end{align}
where
\begin{align*}
v_X(x) &\coloneqq [P_X^\top K_X^{-1} P_X]^{-1} [k_X(x) K_X^{-1} P_X  - p(x)]^\top , \\
u_X(x) &\coloneqq K_X^{-1} [ k_X(x)^\top - P_X v_X(x) ]^\top
\end{align*}
are obtained from the solution of the invertible linear system
\begin{equation*}
\left[ \begin{array}{cc} K_X & P_X \\ P_X^\top & 0 \end{array} \right] \left[ \begin{array}{c} u_{X}(x) \\ v_{X}(x) \end{array} \right] = \left[ \begin{array}{c} k_X(x)^\top \\ p(x)^\top \end{array} \right]
\end{equation*}
and have the cardinality properties $[u_X(x_j)]_i = \delta_{ij}$ and $[v_X(x_j)]_i = 0$ for every $i,j \in \{ 1,\ldots,n \}$.
\end{theorem}

The proof is similar to that of Thm.\ \ref{supplement-thm:LagrangeForm} and is therefore omitted.
Note the reversal in the roles of $u_{X,\Sigma}$ and $v_X$ in Eqns.\ \ref{supplement-eq:CovLagrange} and \ref{supplement-eqn:FlatPriorCovLagrange} for the posterior covariance.

\subsection{Results on Cubature}

\begin{customtheorem}{\ref{thm:BQLimit}}[Bayes--Sard cubature; BSC]
Consider the Gaussian process $f \mid \mathcal{D}_X$ defined in Thm.\ \ref{thm:PosteriorMean} and suppose that $X$ is a $\pi$-unisolvent point set.
Then in the limit $\Sigma^{-1} \rightarrow 0$ we have that \sloppy{${I(f) \mid \mathcal{D}_X \sim \mathcal{N}( \mu_X(f^\dagger) , \sigma_X^2 )}$} with
\begin{equation*}
\mu_X(f^\dagger) = w_k^\top f_X^\dagger - w_\pi^\top \eta \hspace{0.5cm} \text{and} \hspace{0.5cm} \sigma_X^2 = k_{\nu,\nu} - k_{\nu,X} K_X^{-1} k_{\nu,X}^\top + \big( k_{\nu,X} K_X^{-1} P_X - p_\nu \big) w_\pi,
\end{equation*}
where the weight vectors $w_{k} \in \R^n$ and $w_{\pi} \in \R^Q$ are obtained from the solution of the linear system
\begin{equation}\label{supplement-eqn:FlatPriorWeights}
\left[ \begin{array}{cc} K_X & P_X \\ P_X^\top & 0 \end{array} \right] \left[ \begin{array}{c} w_{k} \\ w_{\pi} \end{array} \right]  =  \left[ \begin{array}{c} k_{\nu,X}^\top \\ p_\nu^\top \end{array} \right].
\end{equation}
Equivalently, $w_{k} = I(u_{X})$ and $w_{\pi} = I(v_{X})$ for the Lagrange functions of Thm.\ \ref{supplement-thm:FlatPriorLagrange}.
\end{customtheorem}

\begin{proof}
As we have only established that $s_{X,\Sigma}(f^\dagger) \to s_X(f^\dagger)$ and $k_{X,\Sigma} \to k_X$ pointwise in Thm.\ \ref{thm:LimitInterpolant}, we cannot directly deduce that
\begin{align*}
\mu_{X,\Sigma}(f^\dagger) & \to \mu_X(f^\dagger) = \int_D s_X(f^\dagger)(x) \mathrm{d}\nu(x), \\
\sigma_{X,\Sigma}^2 & \to \sigma_X^2 = \int_D \int_D k_X(x,x') \mathrm{d}\nu(x) \mathrm{d} \nu(x').
\end{align*}
However, that this is indeed the case can be confirmed by carrying out analysis analogous to that in the proof Thm.\ \ref{thm:LimitInterpolant}, based on Neumann series, for $\mu_{X,\Sigma}(f^\dagger)$ and $\sigma_{X,\Sigma}^2$ at the limit $\Sigma^{-1} \to 0$.
To avoid repetition, the details are omitted.
\end{proof}

\begin{customtheorem}{\ref{thm:PolyReproBC}}
Suppose that $\dim(\pi) = n$ and let $X$ be a $\pi$-unisolvent set. If $\eta = 0$, then
\begin{equation*}
\mu_X(f^\dagger) = w_k^\top f_X^\dagger, \hspace{0.5cm}w_k^\top = p_\nu P_X^{-1}, \hspace{0.5cm} \mu_X(p) = I(p) \hspace{0.5cm} \text{for every $p \in \pi$} \hspace{0.5cm} 
\end{equation*}
and
\begin{equation*}
\sigma_X^2 = e_k(X,w_k)^2 = k_{\nu,\nu} - 2 k_{X,\nu} w_k + w_k^\top K_X w_k.
\end{equation*}
That is, the Bayes--Sard cubature weights $w_k$ are the unique weights such that every function in $\pi$ is integrated exactly and the posterior standard deviation $\sigma_X$ coincides with the WCE in the RKHS $H(k)$.
\end{customtheorem}

\begin{proof}
Due to $\dim(\pi) = n$ and $X$ being a $\pi$-unisolvent set, the Vandermonde matrix $P_X$ is an invertible square matrix. From Eqn.\ \ref{supplement-eqn:FlatPriorWeights} we have
\begin{equation*}
\begin{split}
w_k &= \big( K_X^{-1} - K_X^{-1} P_X [P_X^\top K_X^{-1} P_X]^{-1} P_X^\top K_X^{-1} \big) k_{\nu,X}^\top + K_X^{-1} P_X [P_X^\top K_X^{-1} P_X]^{-1} p_\nu^\top \\
&= P_X^{-\top} p_\nu^\top.
\end{split}
\end{equation*}
These are the unique weights satisfying $\sum_{j=1}^n w_{k,j} p_i(x_j) = I(p_i)$ for each basis function $p_i$ of $\pi$.
Similarly, the weights $w_\pi$ take the form
\begin{equation*}
w_\pi = [P_X^\top K_X^{-1} P_X]^{-1} P_X^\top K_X^{-1} k_{\nu,X}^\top - [P_X^\top K_X^{-1} P_X]^{-1} p_\nu^\top = P_X^{-1} k_{\nu,X}^\top - [P_X^\top K_X^{-1} P_X]^{-1} p_\nu^\top,
\end{equation*}
so that
\begin{equation*}
\begin{split}
\sigma_X^2 &= k_{\nu,\nu} - k_{\nu,X} K_X^{-1} k_{\nu,X}^\top + \big( P_X^\top K_X^{-1} k_{\nu,X}^\top - p_\nu^\top \big)^\top w_\pi \\
&= k_{\nu,\nu} - k_{\nu,X} K_X^{-1} k_{\nu,X}^\top + \big( P_X^\top K_X^{-1} k_{\nu,X}^\top - p_\nu^\top \big)^\top \big( P_X^{-1} k_{\nu,X}^\top - [P_X^\top K_X^{-1} P_X]^{-1} p_\nu^\top \big) \\
&= k_{\nu,\nu} - 2 k_{\nu,X} w_k + w_k^\top K_X w_k.
\end{split}
\end{equation*}
We recognise this final expression as the squared worst-case error from Eqn.\ \ref{supplement-eqn:wceExplicit}.
\end{proof}

\begin{customcorollary}{\ref{cor:cubatureRepro}}
Consider an $n$-point cubature rule with points $X$ and non-zero weights $w \in \R^n$ and assume that $\nu$ admits a positive density function and that $\eta = 0$.
Then there is a function space $\pi$ of dimension $n$, such that the Bayes--Sard method recovers $w_k = w$ and $\sigma_X^2 = e_k(X, w)^2$, as defined in Thm.\ \ref{thm:BQLimit}.
\end{customcorollary}

\begin{proof}
From the assumption that $\nu$ has a positive density function with respect to the Lebesgue measure it follows that $\nu(\{x\}) = 0$ for every $x \in D$ and that for any distinct points $x_1,\ldots,x_n \in D$ there exist disjoint sets $D_i$ of positive measure such that $x_i \in D_i$.
Select then the $n$ functions
\begin{equation*}
p_i = \mathbbm{1}_{D_i \setminus \{x_i\}} + \frac{\nu(D_i)}{w_i} \mathbbm{1}_{\{x_i\}}.
\end{equation*}
It holds that $I(p_i) = \nu(D_i)$.
The associated Vandermonde matrix is diagonal and has the elements $[P_X]_{ii} = \nu(D_i)/w_i$.
Hence it can be trivially inverted.
It follows that the Bayes--Sard method with basis $\{p_1,\dots,p_n\}$ has a posterior mean $\mu_X(f^\dagger) = w^\top f_X^\dagger$.
\end{proof}

The construction is more appealing if the weights are positive and their sum does not exceed one, since then we can use $p_i = \mathbbm{1}_{D_i}$ for disjoint sets such that $\nu(D_i) = w_i$ and $x_i \in D_i$, or if the weights are naturally given by exactness conditions on $\pi$ and $X$ is $\pi$-unisolvent.
Examples of such more natural constructions include uniformly weighted (quasi) Monte Carlo rules, that arise from using a partition $D = \cup_{i=1}^n D_i$ with $\nu(D_i) = 1/n$, and Gaussian tensor product rules.

\begin{remark}[All cubature rules are Bayes rules for some prior]
If we recall that the posterior mean is a Bayes decision rule for squared-error loss \cite{Berger2013a} then the above corollary demonstrates that ``any cubature rule is a Bayes decision rule for some prior''; a concrete instantiation of the \emph{complete class} theorem of Wald~\cite{Wald1947}.
\end{remark}

\subsection{Details on Assumptions in Thm.\ \ref{thm:maternConvergence}} \label{subsec:maternDetails}

This section collects the assumptions required in Thm.\ \ref{thm:maternConvergence}.

\begin{definition}[Interior cone condition]
A bounded domain $D \subset \R^d$ is said to satisfy an \emph{interior cone condition} if there exists an angle $\theta \in (0,\frac{\pi}{2})$ and a radius $r > 0$ such that for each $x \in D$ a unit vector $\xi(x)$ exists such that the cone $\{x + \lambda y : y \in \mathbb{R}^d, \|y\|_2 = 1, y^\top \xi(x) \geq \cos \theta, \lambda \in [0,r]\}$ is contained in $D$.
\end{definition}

\begin{definition}[Norm-equivalence]
Two norms $\norm{\cdot}_1$ and $\norm{\cdot}_2$ on a vector space $H$ are said to be \emph{equivalent} if there exist positive constants $C_1$ and $C_2$ such that $C_1 \norm{v}_1 \leq \norm{v}_2 \leq C_2 \norm{v}_1$ for every $v \in V$.
\end{definition}

The boundary of a bounded open domain $D$ is said to be \emph{Lipschitz} if it is ``regular enough''.
See for example~\mbox{\cite[Sec.\ 3]{Kanagawa2017}} for a formal definition.
Most domains of interest to us, such as convex sets, have a boundary that is Lipschitz.

\subsection{Explicit Rates of Convergence (for the case $Q = n$)} \label{subsec:Q=n-rates}

As pointed out in Cor.\ \ref{cor:cubatureRepro}, the mean $\mu_X(f^\dagger)$ of the Bayes--Sard output can be arranged to coincide with any given cubature rule through judicious choice of the function space $\pi$, provided that its dimension matches the number of nodes $x_i$ that are used.
In this case, convergence rates are trivially inherited.
For example, and for simplicity letting $\nu$ be uniform on $D = [0,1]^d$,
\begin{itemize}
\item nodes drawn randomly (or through utilisation of a Markov chain) from $\nu$ and uniform weights yield the standard (probabilistic) Monte Carlo rate
\begin{equation*}
\mathbb{E} \big( \abs[\big]{ \mu_X^\text{\tiny{MC}}(f^\dagger) - I(f^\dagger) }^2 \big)^{\frac{1}{2}} \lesssim n^{-1/2} \|f^\dagger\|_{L^2(D)} ;
\end{equation*}
\item for $\alpha \geq 2$, certain quasi-Monte Carlo methods can attain polynomial rates for functions in the space $H_{\text{mix}}^\alpha(D)$ of dominating mixed smoothness:
\begin{equation*}
\abs[\big]{ \mu_X^\text{\tiny{QMC}}(f^\dagger) - I(f^\dagger) } \lesssim n^{-\alpha+\varepsilon} \|f^\dagger\|_{H_{\text{mix}}^\alpha(D)}
\end{equation*}
for any $\varepsilon > 0$.
See~\cite[Chapter 15]{Dick2010} for these results and for the formal definition of the norm;
\item certain sparse grid methods on hypercubes have the rates
\begin{align*}
\abs[\big]{ \mu_X^\text{\tiny{SG}}(f^\dagger) - I(f^\dagger) } & \lesssim n^{-\alpha/d} (\log n)^{(d-1)(\alpha/d + 1)} \, \norm{f^\dagger}_{C^\alpha(D)}, \\
\abs[\big]{ \mu_X^\text{\tiny{SG}}(f^\dagger) - I(f^\dagger) } & \lesssim  n^{-\alpha} (\log n)^{(d-1)(\alpha + 1)} \, \norm{f^\dagger}_{F^\alpha(D)}
\end{align*}
for functions having bounded derivatives or bounded mixed derivatives up to order $\alpha$, respectively. See~\cite{NovakRitter1996,NovakRitter1997,NovakRitter1999} for these results and for formal definitions of the norms.
\end{itemize}


\section{Unisolvent Point Sets} \label{sec:unisolvency}

This section contains more details and examples about unisolvent point sets.

\begin{customdefinition}{\ref{def:Unisolvent}}[Unisolvency]
Let $\pi$ denote a finite-dimensional linear subspace of real-valued functions on $D$.
A point set $X = \{x_1,\ldots,x_n\} \subset D$ with $n \geq \dim(\polspace)$ is called $\polspace$-\emph{unisolvent} if the zero function is the only element in $\polspace$ that vanishes on $X$.
\end{customdefinition}

The following proposition provides an equivalent operational characterisation of unisolvency:

\begin{proposition} \label{supplement-prop:unisolvent}
Let $\{p_1,\dots,p_Q\}$ denote a basis of $\pi$, so that $Q = \dim(\pi)$.
Then a point set $X$ is $\polspace$-unisolvent if and only if the $n \times Q$ Vandermonde matrix $P_X$ is of full rank. 
\end{proposition}

\begin{example}[Cartesian product of a unisolvent set]
As a simple example of how one can generate a unisolvent set in $\mathbb{R}^d$, consider the Cartesian grid $X = Z^d$ for a $\Pi_{m-1}(\mathbb{R})$-unisolvent set \sloppy{${Z = \{z_1,\ldots,z_m\}\subset \R}$} (i.e., the points $z_i$ are distinct).
Then for any $d$-variate polynomial
\begin{equation*}
p \in \Pi \coloneqq \mathrm{span} \{ x^\alpha \colon \alpha \leq m-1 \},
\end{equation*}
the univariate polynomial
\begin{equation*}
p_j(z) = p(z_{\alpha_1},\ldots,z_{\alpha_{j-1}},z,z_{\alpha_{j+1}},\ldots,z_{\alpha_d})
\end{equation*}
is of degree at most $m-1$ and, for any indices $j \in \{1,\ldots,d\}$ and $\alpha_1,\ldots,\alpha_d \in \{1,\ldots,m-1\}$, the polynomial $p_j$ cannot vanish on $Z$ unless it is the zero polynomial. 
It follows that $p$ cannot vanish on $X$ unless $p \equiv 0$.
Therefore $X$ is $\Pi$-unisolvent. 
Note that $\#X = \dim(\Pi) = m^d$.
\end{example}

\begin{example}[Not all sets are unisolvent]
As a counterexample, consider six points \sloppy{${X = \{(x_i,y_i), i=1,\ldots,6\}}$} on a unit circle in $\R^2$.
These points are not $\Pi_2(\R^d)$-unisolvent: the associated Vandermonde matrix
\begin{equation*}
P_X = 
\left[
\begin{array}{cccccc}
1 & x_1 & y_1 & x_1 y_1 & x_1^2 & y_1^2 \\
1 & x_2 & y_2 & x_2 y_2 & x_2^2 & y_2^2 \\
1 & x_3 & y_3 & x_3 y_3 & x_3^2 & y_3^2 \\
1 & x_4 & y_4 & x_4 y_4 & x_4^2 & y_4^2 \\
1 & x_5 & y_5 & x_5 y_5 & x_5^2 & y_5^2 \\
1 & x_6 & y_6 & x_6 y_6 & x_6^2 & y_6^2 \\
\end{array}
\right]
\end{equation*}
for the canonical polynomial basis is not of full rank as the first column is the sum of the last two columns.
\end{example}

Intuitively, ``almost all'' point sets are unisolvent, but to actually verify that an arbitrary point set $X$ is unisolvent, from Proposition \ref{supplement-prop:unisolvent} it is required to compute the rank of the Vandermonde matrix $P_X$, which entails a super-linear computational cost~\cite{SauerXu1995}.
However, certain point sets are guaranteed to be unisolvent:
\begin{itemize}
\item When $\polspace$ is a Chebyshev system (so that its basis functions are so-called \emph{generalised polynomials}) in one dimension, any set $X \subset \R$ of distinct points is $\pi$-unisolvent~\cite{KarlinStudden1966}.
\item For $\polspace$ spanned by the indicator functions $\mathbbm{1}_{A_1},\ldots,\mathbbm{1}_{A_n}$ of disjoint sets $A_i \subset D$ such that $x_i \in A_i$, the set $X$ is $\pi$-unisolvent and $P_X$ is the $n \times n$ identity matrix.
\item \emph{Padua points} on $[-1, 1]^2$ are known to be unisolvent with respect to polynomial spaces~\cite{Caliari2005}.
\item Recent algorithms for generating moderate number of points for polynomial interpolation, with a unisolvency guarantee on the output, can be used~\cite{SauerXu1995,Gunzburger2014}.
\end{itemize}


\section{An Equivalent Kernel Perspective} \label{sec:kernel-perspective}

In this section we interpret the output of the Bayes--Sard method from the perspective of the reproducing kernel, in order to provide additional insight that complements the main text.
The formulation of cubature rules in reproducing kernel Hilbert spaces dates back to~\cite{Larkin1970,Richter1970,RichterDyn1971a,RichterDyn1971b,Larkin1972} and in particular the integrated kernel interpolant was studied in~\cite{Bezhaev1991} and \cite{SommarivaVianello2006}.

\subsection{Interpolation}\label{sec:equivalent-interpolation}

There is a well-understood equivalence between Gaussian process regression and optimal interpolation in reproducing kernel Hilbert spaces:
Let $\{p_1,\ldots,p_Q\}$ be a basis for $\pi$ and define the kernel $k_\pi(x,x') = \sum_{i=1}^Q p_i(x) p_i(x')$.
Consider the kernel
\begin{equation*}
k_\sigma(x,x') = k(x,x') + \sigma^2 k_\pi(x,x')
\end{equation*}
for $\sigma > 0$.
Then the reproducing kernel Hilbert space induced by $k_\sigma$ corresponds to the set
\begin{equation*}
H(k_\sigma) = \big\{ f + p \: \colon f \in H(k), \, p \in \pi \big\}
\end{equation*}
equipped with a particular $\sigma$-dependent inner product.
It can be shown that the interpolant with minimal norm in $H(k_\sigma)$ is unique and given by 
\begin{equation*}
s_{X,\sigma}(f^\dagger)(x) = [k_X(x) + \sigma^2 k_{\pi,X}(x)]  [K_X + \sigma^2 P_X P_X^\top]^{-1} f_X^\dagger ,
\end{equation*}
where the row vector $k_{\pi,X}(x)$ has the elements $k_\pi(x,x_j)$.
When $\eta = 0$, it is straightforward to show that $s_{X,\sigma}(f^\dagger) = s_{X,\Sigma}(f^\dagger)$ for $\Sigma = \sigma^2 I$ and thus $s_{X,\sigma}(f^\dagger) \to s_X(f^\dagger)$ pointwise as $\sigma \to \infty$.
The kernel interpolation operator $s_X$ is well-studied and the reader is referred to, for example, Sec.\ 8.5 of~\cite{Wendland2005}.

\subsection{Cubature}

The worst-case error $e_k(X,w)$ of a cubature rule described by the points $X = \{x_1,\ldots,x_n\} \subset D$ and weights $w = (w_1,\ldots,w_n) \in \R^n$ has the explicit form
\begin{equation}\label{supplement-eqn:wceExplicit}
e_k(X,w) \coloneqq \sup_{\norm{h}_k \leq 1} \, \abs[\bigg]{\, \sum_{i=1}^n w_i h(x_i) - \int_D h \mathrm{d}\nu \, } = \big( k_{\nu,\nu} - 2 k_{\nu,X} w + w^\top K_X w \big)^{1/2}.
\end{equation}
See for example~\cite[Cor.\ 3.6]{Oettershagen2017}.

Recall from Sec.\ \ref{sec:CubatureRepro} that the weights $w_\text{BC}$ of the standard Bayesian cubature rule are worst-case optimal in the reproducing kernel Hilbert space $H(k)$ induced by the kernel~$k$:
\begin{equation*}
w_\text{BC} = \argmin_{w \in \R^n} e_k(X,w).
\end{equation*}
Conveniently, the minimum corresponds to the integration error for the kernel mean function $k_\nu$ which acts as the representer of integration (i.e., $\langle h, k_\nu \rangle_k = I(h)$ for $h \in H(k)$):
\begin{equation*}
e_k(X, w_\text{BC}) = \big( k_{\nu,\nu} - k_{\nu,X} w_\text{BC} \big)^{1/2}.
\end{equation*}

Now, turning to Bayes--Sard cubature, we have from Sec.\ \ref{sec:equivalent-interpolation} that the Bayes--Sard cubature rule $\mu_X(f^\dagger)$ can be cast as an optimal cubature method based on the kernel
\begin{equation*}
k_\sigma(x,x') = k(x,x') + \sigma^2 k_\pi(x,x')
\end{equation*}
in the $\sigma \to \infty$ limit.
The following therefore holds for the weights $w_k$ and variance $\sigma_X^2$ of the Bayes--Sard cubature method:
\begin{equation*}
w_k = \lim_{\sigma \to \infty} \argmin_{w \in \R^n} e_{k_\sigma}(X, w), \qquad
\sigma_X^2 = \lim_{\sigma \to \infty} \, \min_{w \in \R^n} e_{k_\sigma}(X, w)^2.
\end{equation*}
Recall that $H(k_\sigma)$ consists of functions which can be expressed as sums of elements of $H(k)$ and~$\pi$. 
To simplify the following argument, assume $f^\dagger \in H(k_\sigma)$.
That the elements of $\pi$ are exactly integrated can be clearly understood in this context.
Indeed, the norm of a function $h \in H(k_\sigma)$ is~\mbox{\cite[Sec.\ 4.1]{Berlinet2011}}
\begin{equation*}
\norm{h}_{k_\sigma}^2 = \min_{g \in H(k), \, p \in \pi} \big\{ \norm{g}_k^2 + \sigma^2 \norm{p}_{k_\pi}^2 \: \colon g + p = h \big\}.
\end{equation*}
Thus, in terms of function approximation, when $\sigma \rightarrow \infty$ the error $\|f^\dagger - h\|_{k_\sigma}$ is dominated by the error $\sigma^2 \| \mathrm{P}_\pi(f^\dagger) - \mathrm{P}_\pi(h)\|_{k_\pi}$ where $\mathrm{P}_\pi$ is the orthogonal projection onto $\pi$ in $H(k_\sigma)$.
Thus, under this norm, the approximation of $\mathrm{P}_\pi(f^\dagger)$ in $\pi$ is prioritised.
In particular, when $\dim(\pi) = n$, the weights $w_k$ are fully-determined by the requirement of exactness for functions in $\pi$ and nothing is done to integrate functions in $H(k)$ well. 
Consequently, the limiting variance $\sigma_X^2$ must coincide with the (squared) worst-case error $e_k(X, w_k)^2$ in the RKHS $H(k)$.

\begin{remark} \label{rmk:devore}
Alternatively, the limiting weights $w_k$ can be seen as a solution to the constrained convex optimisation problem of minimising the RKHS approximation error to the kernel mean function $k_\nu$ under exactness conditions for functions in $\pi$:
\begin{equation*}
w_k = \argmin_{w \in \R^n} \, \norm{ k_\nu - k_X w }_k \hspace{0.5cm} \text{subject to} \hspace{0.5cm} P_X^\top w = p_\nu^\top.
\end{equation*}
This can be verified in a straightforward manner based on~\cite[Section 5.2]{DeVore2017}.
\end{remark}

\section{Further Details for Numerical Experiments} \label{sec:experiments}

This section contains further details about the zero coupon bonds example of Sec.\ \ref{sec:zcb}.
See~\mbox{\cite[Sec.\ 6.1]{Holtz2011}} for a complete account.

The $d$-step Euler--Maruyama with uniform step-size $\Delta t = T/D$ of the Vasicek model
\begin{equation*}
\mathrm{d}r(t) = \kappa\big( \theta - r(t) \big) \mathrm{d}t + \sigma \mathrm{d} W(t),
\end{equation*}
where $W(t)$ is the standard Brownian motion and $\kappa$, $\theta$, and $\sigma$ are positive parameters, is
\begin{equation*}
r_{t_i} = r_{t_{i-1}} + \kappa\big(\theta - r_{t_{i-1}}\big) \Delta t + \sigma \sqrt{\Delta t} x_{t_i}, \hspace{0.5cm} i=1,\ldots,d,
\end{equation*}
for independent standard Gaussian random variables $x_{t_i}$ and some initial value $r_{t_0}$.
The quantity of interest is the Gaussian expectation
\begin{equation*}
P(0,T) \coloneqq \mathbb{E}\Bigg[ \exp\bigg( -\Delta t \sum_{i=0}^{D-1} r_{t_i} \bigg)\Bigg] = \exp(-\Delta t r_{t_0}) \mathbb{E}\Bigg[ \exp\bigg( -\Delta t \sum_{i=1}^{D-1} r_{t_i} \bigg)\Bigg]
\end{equation*}
of dimension $d \coloneqq D-1$.
This expectation admits the closed-form solution
\begin{equation*}
P(0,T) = \exp\bigg( -\frac{(\gamma + \beta_D r_{t_0}) T}{d}\bigg)
\end{equation*}
for certain constants $\gamma$ and $\beta_D$.
In the experiment, we used the same parameter values as in~\cite{Holtz2011}:
\begin{equation*}
\kappa = 0.1817303, \quad \theta = 0.0825398957, \quad \sigma = 0.0125901, \quad r_0 = 0.021673, \quad T = 5.
\end{equation*}

\newpage

\small
\bibliographystyle{plain}

\begin{thebibliography}{}

\end{thebibliography}


\begin{thebibliography}{10}

\bibitem{Bach2017}
F.~Bach.
\newblock On the equivalence between kernel quadrature rules and random feature
  expansions.
\newblock {\em Journal of Machine Learning Research}, 18(21):1--38, 2017.

\bibitem{Berger2013a}
J.~Berger.
\newblock {\em Statistical Decision Theory: Foundations, Concepts, and
  Methods}.
\newblock Springer Science \& Business Media, 2013.

\bibitem{Berlinet2011}
A.~Berlinet and C.~Thomas-Agnan.
\newblock {\em Reproducing Kernel Hilbert Spaces in Probability and
  Statistics}.
\newblock Springer Science \& Business Media, 2011.

\bibitem{Bezhaev1991}
A.~{\BIBYu}. Bezhaev.
\newblock Cubature formulae on scattered meshes.
\newblock {\em Soviet Journal of Numerical Analysis and Mathematical
  Modelling}, 6(2):95--106, 1991.

\bibitem{Bogachev1998}
V.~I. Bogachev.
\newblock {\em Gaussian Measures}.
\newblock Number~62 in Mathematical Surveys and Monographs. American
  Mathematical Society, 1998.

\bibitem{Briol2017}
F.-X. Briol, C.~J. Oates, M.~Girolami, M.~A. Osborne, and D.~Sejdinovic.
\newblock Probabilistic integration: A role in statistical computation?
\newblock {\em arXiv:1512.00933v6}, 2017.

\bibitem{Caliari2005}
M.~Caliari, S.~De~Marchi, and M.~Vianello.
\newblock Bivariate polynomial interpolation on the square at new nodal sets.
\newblock {\em Applied Mathematics and Computation}, 165(2):261--274, 2005.

\bibitem{Chai2018}
H.~Chai and R.~Garnett.
\newblock An improved {B}ayesian framework for quadrature of constrained
  integrands.
\newblock {\em arXiv:1802.04782}, 2018.

\bibitem{Cockayne2017}
J.~Cockayne, C.~J. Oates, T.~Sullivan, and M.~Girolami.
\newblock {B}ayesian probabilistic numerical methods.
\newblock {\em arXiv:1702.03673v2}, 2017.

\bibitem{Davis2007}
P.~J. Davis and P.~Rabinowitz.
\newblock {\em Methods of Numerical Integration}.
\newblock Courier Corporation, 2007.

\bibitem{DeVore2017}
R.~DeVore, S.~Foucart, G.~Petrova, and P.~Wojtaszczyk.
\newblock Computing a quantity of interest from observational data.
\newblock {\em Preprint}, 2017.

\bibitem{Diaconis1988}
P.~Diaconis.
\newblock Bayesian numerical analysis.
\newblock In {\em Statistical Decision Theory and Related Topics IV}, volume~1,
  pages 163--175. Springer-Verlag New York, 1988.

\bibitem{Dick2010}
J.~Dick and F.~Pillichshammer.
\newblock {\em Digital Nets and Sequences: Discrepancy Theory and Quasi-{M}onte
  {C}arlo Integration}.
\newblock Cambridge University Press, 2010.

\bibitem{Gautschi2004}
W.~Gautschi.
\newblock {\em Orthogonal Polynomials: Computation and Approximation}.
\newblock Numerical Mathematics and Scientific Computation. Oxford University
  Press, 2004.

\bibitem{Gunter2014}
T.~Gunter, M.~A. Osborne, R.~Garnett, P.~Hennig, and S.~J. Roberts.
\newblock Sampling for inference in probabilistic models with fast {B}ayesian
  quadrature.
\newblock In {\em Advances in Neural Information Processing Systems}, pages
  2789--2797, 2014.

\bibitem{Gunzburger2014}
M.~Gunzburger and A.~L. Teckentrup.
\newblock Optimal point sets for total degree polynomial interpolation in
  moderate dimensions.
\newblock {\em arXiv:1407.3291}, 2014.

\bibitem{Hennig2015}
P.~Hennig, M.~A. Osborne, and M.~Girolami.
\newblock Probabilistic numerics and uncertainty in computations.
\newblock {\em Proceedings of the Royal Society of London A: Mathematical,
  Physical and Engineering Sciences}, 471(2179), 2015.

\bibitem{Hensman2018}
J.~Hensman, N.~Durrande, and A.~Solin.
\newblock Variational fourier features for {G}aussian processes.
\newblock {\em Journal of Machine Learning Research}, 18(151):1--52, 2018.

\bibitem{Hickernell1998}
F.~Hickernell.
\newblock A generalized discrepancy and quadrature error bound.
\newblock {\em Mathematics of Computation}, 67(221):299--322, 1998.

\bibitem{Holtz2011}
M.~Holtz.
\newblock {\em Sparse Grid Quadrature in High Dimensions with Applications in
  Finance and Insurance}.
\newblock Number~77 in Lecture Notes in Computational Science and Engineering.
  Springer, 2011.

\bibitem{Kanagawa2016}
M.~Kanagawa, B.~K. Sriperumbudur, and K.~Fukumizu.
\newblock Convergence guarantees for kernel-based quadrature rules in
  misspecified settings.
\newblock In {\em Advances in Neural Information Processing Systems}, pages
  3288--3296, 2016.

\bibitem{Kanagawa2017}
M.~Kanagawa, B.~K. Sriperumbudur, and K.~Fukumizu.
\newblock Convergence analysis of deterministic kernel-based quadrature rules
  in misspecified settings.
\newblock {\em arXiv:1709.00147}, 2017.

\bibitem{KarlinStudden1966}
S.~Karlin and W.~J. Studden.
\newblock {\em {T}chebycheff Systems: With Applications in Analysis and
  Statistics}.
\newblock Interscience Publishers, 1966.

\bibitem{Karvonen2017a}
T.~Karvonen and S.~S{\"a}rkk{\"a}.
\newblock Classical quadrature rules via {G}aussian processes.
\newblock In {\em 27th IEEE International Workshop on Machine Learning for
  Signal Processing}, 2017.

\bibitem{Karvonen2018}
T.~Karvonen and S.~S{\"a}rkk{\"a}.
\newblock Fully symmetric kernel quadrature.
\newblock {\em SIAM Journal on Scientific Computing}, 40(2):A697--A720, 2018.

\bibitem{Kennedy2001}
M.~C. Kennedy and A.~O'Hagan.
\newblock Bayesian calibration of computer models.
\newblock {\em Journal of the Royal Statistical Society: Series B (Statistical
  Methodology)}, 63(3):425--464, 2001.

\bibitem{Larkin1970}
F.~M. Larkin.
\newblock Optimal approximation in {H}ilbert spaces with reproducing kernel
  functions.
\newblock {\em Mathematics of Computation}, 24(112):911--921, 1970.

\bibitem{Larkin1972}
F.~M. Larkin.
\newblock Gaussian measure in {H}ilbert space and applications in numerical
  analysis.
\newblock {\em The Rocky Mountain Journal of Mathematics}, 2(3):379--421, 1972.

\bibitem{Larkin1974}
F.~M. Larkin.
\newblock Probabilistic error estimates in spline interpolation and quadrature.
\newblock In {\em Information Processing 74 (Proceedings of IFIP Congress,
  Stockholm, 1974)}, volume~74, pages 605--609. North-Holland, 1974.

\bibitem{Minka2000}
T.~Minka.
\newblock Deriving quadrature rules from {Gaussian} processes.
\newblock Technical report, Statistics Department, Carnegie Mellon University,
  Nov 2000.

\bibitem{MosamamKent2010}
A.~M. Mosamam and J.~T. Kent.
\newblock Semi-reproducing kernel {H}ilbert spaces, splines and increment
  kriging.
\newblock {\em Journal of Nonparametric Statistics}, 22(6):711--722, 2010.

\bibitem{NovakRitter1996}
E.~Novak and K.~Ritter.
\newblock High dimensional integration of smooth functions over cubes.
\newblock {\em Numerische Mathematik}, 75(1):79--97, 1996.

\bibitem{NovakRitter1997}
E.~Novak and K.~Ritter.
\newblock The curse of dimension and a universal method for numerical
  integration.
\newblock In {\em Multivariate Approximation and Splines}, pages 177--187.
  Birkhäuser Basel, 1997.

\bibitem{NovakRitter1999}
E.~Novak and K.~Ritter.
\newblock Simple cubature formulas with high polynomial exactness.
\newblock {\em Constructive Approximation}, 15(4):499--522, 1999.

\bibitem{Oates2017}
C.~J. Oates, S.~Niederer, A.~Lee, F.-X. Briol, and M.~Girolami.
\newblock Probabilistic models for integration error in the assessment of
  functional cardiac models.
\newblock In {\em Advances in Neural Information Processing Systems}, pages
  109--117, 2017.

\bibitem{Oettershagen2017}
J.~Oettershagen.
\newblock {\em Construction of Optimal Cubature Algorithms with Applications to
  Econometrics and Uncertainty Quantification}.
\newblock PhD thesis, Institut f\"{u}r Numerische Simulation, Universit\"{a}t
  Bonn, 2017.

\bibitem{OHagan1978}
A.~O'Hagan.
\newblock Curve fitting and optimal design for prediction.
\newblock {\em Journal of the Royal Statistical Society. Series B
  (Methodological)}, 40(1):1--42, 1978.

\bibitem{OHagan1991}
A.~O'Hagan.
\newblock Bayes--{H}ermite quadrature.
\newblock {\em Journal of Statistical Planning and Inference}, 29(3):245--260,
  1991.

\bibitem{Osborne2012}
M.~Osborne, R.~Garnett, Z.~Ghahramani, D.~K. Duvenaud, S.~J. Roberts, and C.~E.
  Rasmussen.
\newblock Active learning of model evidence using {B}ayesian quadrature.
\newblock In {\em Advances in Neural Information Processing Systems}, pages
  46--54, 2012.

\bibitem{Owhadi2015}
H.~Owhadi and C.~Scovel.
\newblock Conditioning {G}aussian measure on {H}ilbert space.
\newblock {\em arXiv:1506.04208}, 2015.

\bibitem{Rasmussen2006}
C.~E. Rasmussen and C.~K.~I. Williams.
\newblock {\em Gaussian Processes for Machine Learning}.
\newblock MIT Press, 2006.

\bibitem{Richter1970}
N.~Richter.
\newblock Properties of minimal integration rules.
\newblock {\em SIAM Journal on Numerical Analysis}, 7(1):67--79, 1970.

\bibitem{RichterDyn1971b}
N.~Richter-Dyn.
\newblock Minimal interpolation and approximation in {H}ilbert spaces.
\newblock {\em SIAM Journal on Numerical Analysis}, 8(3):583--597, 1971.

\bibitem{RichterDyn1971a}
N.~Richter-Dyn.
\newblock Properties of minimal integration rules. {II}.
\newblock {\em SIAM Journal on Numerical Analysis}, 8(3):497--508, 1971.

\bibitem{Sard1949}
A.~Sard.
\newblock Best approximate integration formulas; best approximation formulas.
\newblock {\em American Journal of Mathematics}, 71(1):80--91, 1949.

\bibitem{Sarkka2016}
S.~S{\"a}rkk{\"a}, J.~Hartikainen, L.~Svensson, and F.~Sandblom.
\newblock On the relation between {G}aussian process quadratures and
  sigma-point methods.
\newblock {\em Journal of Advances in Information Fusion}, 11(1):31--46, 2016.

\bibitem{SauerXu1995}
T.~Sauer and Y.~Xu.
\newblock On multivariate {L}agrange interpolation.
\newblock {\em Mathematics of Computation}, 64(211):1147--1170, 1995.

\bibitem{Schoenberg1964}
I.~J. Schoenberg.
\newblock Spline interpolation and best quadrature formulae.
\newblock {\em Bulletin of the American Mathematical Society}, 70(1):143--148,
  1964.

\bibitem{SommarivaVianello2006}
A.~Sommariva and M.~Vianello.
\newblock Numerical cubature on scattered data by radial basis functions.
\newblock {\em Computing}, 76(3--4):295--310, 2006.

\bibitem{Wald1947}
A.~Wald.
\newblock An essentially complete class of admissible decision functions.
\newblock {\em The Annals of Mathematical Statistics}, pages 549--555, 1947.

\bibitem{Wendland2005}
H.~Wendland.
\newblock {\em Scattered Data Approximation}, volume~28 of {\em Cambridge
  Monographs on Applied and Computational Mathematics}.
\newblock Cambridge University Press, 2005.

\bibitem{Xu2017}
W.~Xu and M.L. Stein.
\newblock Maximum likelihood estimation for a smooth {G}aussian random field
  model.
\newblock {\em SIAM/ASA Journal on Uncertainty Quantification}, 5(1):138--175,
  2017.

\bibitem{ZhouEtAl2018}
Q.~Zhou, W.~Liu, J.~Li, and Y.~M. Marzouk.
\newblock An approximate empirical {B}ayesian method for large-scale
  linear-{G}aussian inverse problems.
\newblock {\em arXiv:1705.07646v2}, 2018.

\end{thebibliography}

\providecommand{\BIBYu}{Yu}

\end{document}